\documentclass[a4paper,12pt,oneside,reqno]{amsart}
\usepackage{amssymb,amsfonts,amsmath,amsthm,amscd, bbm}
\usepackage{graphicx, color}
\usepackage{mathrsfs}
\usepackage [pdfstartview=FitH]{hyperref}
\newcommand{\beq}{\begin{equation}} 
\newcommand{\eeq}{\end{equation}} 
\newcommand{\bea}{\begin{aligned}}
\newcommand{\eea}{\end{aligned}}
\newcommand{\bdm}{\begin{displaymath}}
\newcommand{\edm}{\end{displaymath}}
\newcommand{\barr}{\begin{array}}
\newcommand{\earr}{\end{array}}
\newcommand{\ben}{\begin{enumerate}}
\newcommand{\een}{\end{enumerate}}
\newcommand{\bde}{\begin{description}}
\newcommand{\ede}{\end{description}}

\newcommand{\oball}{B_{\boldsymbol \nu, r}}

\newcommand{\cball}{\overline{B_{\boldsymbol \nu, r}}}
\newcommand{\cv}[1]{\mathcal{A}_N\left(#1\right)}
\newcommand{\cvsquared}[1]{\mathcal{A}^{\,2}_N\left(#1\right)}

\newcommand{\smot}{{\small \otimes}}

\newcommand{\repav}[1]{\E \left\langle \, #1 \, \right\rangle_N}
\newcommand{\repave}[1]{\E \left\langle \, #1 \, \right\rangle_{N;\varepsilon}}

\newcommand{\qap}{q_N(\alpha, \alpha')}

\newtheorem{teor}{Theorem}[section]
\newtheorem{prop}[teor]{Proposition} 
\newtheorem{lem}[teor]{Lemma}

\newtheorem{Def}[teor]{Definition}

\newtheorem{rem}[teor]{Remark}
\newcommand{\defi}{\equiv} 
\newcommand{\R}{\mathbb{R}}

\newcommand{\N}{\mathbb{N}}

\newcommand{\PP}{\mathbb{P}}
\newcommand{\E}{{\mathbb{E}}}
\newcommand{\EE}{{\sf{E}}}

\newcommand{\Hm}{H}

\newcommand{\defeq}{\stackrel{\text{def}}{=}}

\newcommand{\be}{\beta}

\newcommand{\s}{\sigma}
\newcommand{\vare}{\varepsilon}

\newcommand{\G}{\boldsymbol \Phi}
\newcommand{\Gm}{{\mathcal G}_N(\alpha)}
\newcommand{\Gmp}{{\mathcal G}_N(\alpha')}
\newcommand{\Gmt}{\tilde {\mathcal G}_N(\alpha)}
\newcommand{\Gmtp}{\tilde {\mathcal G}_N(\alpha')}
\newcommand{\Pa}{\emph{P}}

\newcommand{\ps}{\mathcal{M}_1^+(S)}

\newcommand{\gqm}{\boldsymbol \nu^{q,m}}
\newcommand{\gqmb}{\boldsymbol \nu^{\bar q, \bar m}}
\newcommand{\gqmu}{\nu^{q,m}_1}
\newcommand{\gqmub}{\nu^{\bar q, \bar m}_1}

\newcommand{\bq}{\overline{q}}
\newcommand{\bam}{\overline{m}}

\newcommand{\om}{\omega_N}
\newcommand{\vecb}[2]{\left( \barr{c} #1 \\ #2 \earr \right)}
\newcommand{\vect}{\bar{\boldsymbol{v}}}

\newcommand{\psd}{\mathcal{M}^+_1\left(S^d\right)}
\newcommand{\pst}{\mathcal{M}^+_1\left(S^2\right)}

\newcommand{\Lb}{\boldsymbol L_{N,\alpha}}
\newcommand{\Tlb}{\tilde{\boldsymbol L}_{N,\alpha}}
\newcommand{\Tlbp}{\tilde{\boldsymbol L}_{N,\alpha'}}
\newcommand{\tlbu}{\tilde L^1_{N,\alpha}}
\newcommand{\tlbd}{\tilde L^2_{N,\alpha}}

\newcommand{\Lbu}{ L_{N,1}}

\newcommand{\nuop}{\bar{\boldsymbol \nu}}
\newcommand{\Phib}{\boldsymbol \Phi}

\newcommand{\nub}{\boldsymbol \nu}
\newcommand{\nubb}{\bar{\boldsymbol \nu}}
\newcommand{\rhob}{\boldsymbol \rho}

\newcommand{\snk}{\bar{\boldsymbol \nu}_{n_k}}
\newcommand{\xib}{{\boldsymbol \xi}}

\newcommand{\mC}{\mathcal{C}}
\newcommand{\mK}{\mathcal{K}}

\newcommand{\arth}{\beta(\sqrt{q_0}u+\sqrt{q_1-q_0}v+h)}

\begin{document}

\title[REM approximation of TAP free energies]{On the REM approximation of TAP free energies.}

\author[N. Kistler]{Nicola Kistler}
\address{Nicola Kistler \\ J.W. Goethe-Universit\"at Frankfurt, Germany.}
\email{kistler@math.uni-frankfurt.de}

\author[M.A. Schmidt]{Marius A. Schmidt}
\address{Marius A. Schmidt  \\ J.W. Goethe-Universit\"at Frankfurt, Germany.}
\email{M.Schmidt@mathematik.uni-frankfurt.de}

\author[G. Sebastiani]{Giulia Sebastiani}
\address{Giulia Sebastiani \\ J.W. Goethe-Universit\"at Frankfurt, Germany.}
\email{sebastia@math.uni-frankfurt.de}

\thanks{This work has been supported by a DFG research grant, contract number 2337/1-1.  NK wishes to express his gratitude to Anton Bovier for constant support and interest in this line of research. }


\date{\today}

\begin{abstract}
The free energy of TAP-solutions for the SK-model of mean field spin glasses can be expressed as a nonlinear functional of local terms: we exploit this feature in order to contrive abstract REM-like models which we then solve by a classical large deviations treatment. This allows to identify the origin of the physically unsettling quadratic (in the inverse of temperature) correction to the Parisi free energy for the SK-model, and formalizes the \emph{true} cavity dynamics which acts on TAP-space, i.e. on the space of TAP-solutions.  From a non-spin glass point of view,  this work is the first in a series of refinements which addresses the stability of hierarchical structures in models of evolving populations. 
\end{abstract}

\maketitle

\section{Introduction}
The Generalized Random Energy Models,  GREM for short, are toy models for mean field spin glasses introduced by Derrida in the 1980's  \cite{Dgrem},  which have played a key role in our understanding of certain aspects of the Parisi theory \cite{MPV}.    Notwithstanding,  the deeper relation between the GREMs and more realistic spin glasses such as the prototypical Sherrington-Kirkpatrick model \cite{sk}, SK for short,  hasn't yet been identified: the goal of this paper is to fill this gap.

Precisely,  we relate the simplest of Derrida's models, the REM \cite{D}, and the {Thouless-Anderson-Palmer free energies} \cite{TAP},  TAP for short; this seamlessly leads to abstract, and what is crucial:  \emph{highly nonlinear},  REM-like Hamiltonians involving only the alleged geometrical properties of the (relevant) TAP solutions, which we then solve within Boltzmann formalism by means of a classical, Sanov-type large deviation analysis.  

For these abstract models we furthermore derive a dual, Parisi-like formula for the free energy,   establish  their convergence to the Derrida-Ruelle cascades \cite{R},  and show that the abstract overlap concentrates on two possible values only -- in complete agreement with the Parisi theory for models within the 1-step replica symmetry breaking (1RSB) approximation.  The inherent nonlinearities also shed new light on the nature of the Parisi formula for the SK-model,  see Section \ref{link_parisi} below. 


What is perhaps more,  our findings {\it i)} considerably improve and clarify the cavity approach \cite{MPV} to mean field spin glasses put forward by 
Aizenman, Ruzmaikina and Arguin \cite{AR, AA}, as well as by Bolthausen and the first author \cite{BK1, BK2}; {\it ii)}
provide a first\footnote{partial, due to the REM-assumption made here.} answer to the question raised in \cite[p.109]{B} concerning the link between the Bolthausen-Sznitman abstract cavity set-up \cite{A} and the SK-model.  In both cases,  headway is made by enhancing the framework of these papers with the missing ingredient "TAP-free energy" : this simple,  yet far-reaching insight is arguably the main contribution of this work. \\

\noindent This paper is organised as follows: in Section \ref{sktapp} we recall some of the main aspects from the picture canvassed in \cite{TAP}. This will motivate and justify our abstract REM-like models which are introduced in Section \ref{model_new}, where the main results are also presented. The proofs are given in the fourth section, with some useful (technical) facts being recalled in the Appendix for the reader's convenience. 

\section{SK, TAP, and Plefka.} \label{sktapp}
The SK-model is the archetypical mean field spin glass: for $N \in \N$,  consider centered Gaussians $( g_{ij})_{1 \leq i < j \leq N}$ issued on some probability space $(\Omega, \mathcal F, \PP)$. These Gaussians, {\it the disorder}, are assumed to be all independent and with variance $1/N$. The SK-Hamiltonian, defined on the Ising configuration space is  then
\beq \label{SK_Ham}
\s \in \Sigma_N \defeq \{\pm 1\}^N \mapsto H_N(\s) = \sum_{1 \leq i < j \leq N} g_{ij} \s_i \s_j \,.
\eeq
The quenched SK-free energy to  inverse temperature $\be>0$ and external field $h \in \R$ is
\beq
N f_N(\be, h) \defi \log \sum_{\s \in \Sigma_N} \exp\left( \be H_N(\s) + h \sum_{i=1}^N \s_i \right)\,.
\eeq
In order to 'solve' the model, Thouless, Anderson and Palmer \cite{TAP} play the delicate (and debatable) card of the spin magnetisation $m_i \defi \left< \s_i \right>_{\be, h, N}, i=1 \dots N$ as order parameter of the theory, with $\left< \right>_{\be, h, N}$ denoting average with respect to the quenched Gibbs measure.  By means of a nonrigorous (and troublesome) diagrammatic expansion,  Thouless {\it et. al.} suggest that the following approximation holds true with overwhelming probability: 
\beq \label{tap_approx}
N f_N(\be, h) = \max_{\boldsymbol m \in \Delta} f_{\text{TAP}}(\boldsymbol m) + o(N) \quad (N\uparrow \infty),
\eeq
where 
\begin{itemize}
\item[i)] the \emph{TAP-free energy} is given by 
\beq \bea \label{tapfe}
Nf_{TAP}(\boldsymbol m) & \defi \beta \sum_{1 \leq i< j \leq N} g_{ij} m_i m_j + h \sum_{i=1}^N m_i \\
 & \qquad \qquad + \frac{\be^2}{4}N\left[ 1- \frac{1}{N}\sum\limits_{i=1}^N m_i^2 \right]^2 - \sum_{i=1}^N I(m_i) \,;
\eea \eeq
\item[ii)] for $m \in [-1,1]$  
\[ 
I(m) \defi \frac{1+m}{2} \log(1+m)+ \frac{1-m}{2} \log(1-m)\,;
\] 
is the classical coin tossing rate function; and
\item[iii)] $\Delta \subset [-1,1]^N$ an unspecified set of restrictions on the quenched magnetisations $\boldsymbol m$.
\end{itemize}

\begin{rem} \label{plefkarem}
Plefka has shown \cite{plefka} that the TAP-approximation \eqref{tap_approx} neatly emerges from a high temperature expansion of the Gibbs potential. As any (finite volume) Gibbs potential, the map $\boldsymbol m \mapsto f_{\text{TAP}}(\boldsymbol m)$ must necessarily be \emph{concave}: in \cite{plefka} it is claimed that this should indeed be the case provided that $\boldsymbol m$ satisfy 
\[
\text{\emph{Plefka's criterium:}}\qquad \frac{\be^2}{N} \sum_{i=1}^N \left(1-m_i^2 \right)^2 < 1 \,.
\]
This condition is widely accepted within the theoretical physics literature (in other words: this restriction should definitely appear in the definition of the $\Delta$-set) but there seems to be divergent opinions if this suffices for the validity of the high temperature expansions and thus of the TAP-approximation \eqref{tapfe}. For a mathematical analysis of the TAP-Plefka approximation within Guerra's interpolation scheme \cite{g}, the reader may check \cite{CPS} and references therein. For an in-depth study of Plefka's convergence criteria for the SK-model, see \cite{GS}.
\end{rem}

Assuming the validity of the TAP-approximation \eqref{tapfe}, we therefore see that extremal "states" must necessarily be critical points of the TAP-free energy: taking the gradient, and rearranging, this leads to the TAP-equations
\beq \label{tapeq}
\nabla f_{\text{TAP}}(\boldsymbol m) = 0 \Longleftrightarrow m_i = \tanh\left( h + \be \sum_{j \neq i} g_{ij} m_j - \be^2(1-q_N(\boldsymbol m)) m_i \right)\,, \; i=1 \dots N, 
\eeq
where $q_N(\boldsymbol m) \defeq (1/N) \sum_{j=1}^N m_j^2$. \\

\noindent In the theoretical physics literature it is claimed that, for large enough $\beta$, the TAP-equations admit exponentially many solutions $\boldsymbol m^{\alpha}, \alpha =1 \dots 2^{\Theta N}$, where $\Theta = \Theta(\be, h)$ is the currently unknown complexity. \\

\noindent Let us now assume to be given a TAP-solution $\boldsymbol m^\alpha$: using \eqref{tapeq} we may express 
\beq
\be \sum_{j \neq i} g_{ij} m_j^\alpha = \tanh^{-1}\left( m_i^\alpha \right) - h + \be^2 \left\{ 1-q_N\left(\boldsymbol m^\alpha \right) \right\}m^\alpha_i.
\eeq
Plugging this into \eqref{tapfe},  and performing some straightforward algebraic manipulations, we obtain a representation of the TAP-FE as a sum of $N$ local terms, as anticipated in the abstract: omitting the elementary details, the upshot reads (by a slight abuse of notation)
\beq \bea
f_{\text{TAP}}\left( \alpha \right) = & \frac{1}{N} \sum_{i}\left[ \frac{1}{2}m_i^\alpha\tanh^{-1}\left( m_i^\alpha \right) + \frac{h}{2} m_i^\alpha -  I(m^\alpha_i) \right] +\frac{\beta^2}{4} \left\{ 1- q_N({\bf m}^\alpha)^2\right\}.
\eea \eeq


What is crucial for our considerations is the nonlinear\footnote{as a matter of fact, quadratic: an analogous expression for the TAP-FE of any $p $-spin model is also available,  in which case the quadratic term turns into a polynomial of degree $p \geq 3$, see e.g. \cite{CLR}.} term in the curly brackets above : this nonlinearity, and only this, will mark the  point of departure from the abstract models studied in \cite{BK1, BK2}. Indeed, introducing the fields
\beq \label{fields}
h_i^\alpha \defi \sum_{j\neq i} g_{ij} m_j^\alpha - \beta \left(1-q_N({\bf m}^{\alpha})\right), 
\eeq
and the associated {\it empirical measures}  
\beq
l_{N, \alpha} \defi \frac{1}{N} \sum_{i=1}^N \delta_{h_i^\alpha},
\eeq
we obtain, through the identity  $I(y)=y\tanh^{-1}(y)-\log\cosh\tanh^{-1}(y)$, the following representation:
\beq \label{tapsamp}
f_{\text{TAP}}\left(\alpha \right) = \Phi \left( \int g(x)^2 l_{N, \alpha}(dx) \right) + \int f_1(x) l_{N, \alpha}(dx)\,,
\eeq
where $\Phi, g, f_1$ are real valued functions given by, respectively: 
\begin{itemize}
\item[SK1)] $x \ni \R \mapsto \Phi(x) \defi \frac{\beta^2}{4} (1-x^2) $;
\item[SK2)] $x \ni \R \mapsto f_1(x) \defi -\frac{1}{2} \log\left( 1- \tanh^2\left(h+\be x \right) \right) - \frac{\be}{2} x \tanh(h+\be x)$;
\item[SK3)] $x \ni \R \mapsto g(x) \defi \tanh(h+\be x) $.
\end{itemize}
The "nonlinear randomness" thus stems from the fluctuations of the {\it self-overlap} 
\beq
q_{\text{EA}}(\alpha) \defeq \int g(x)^2 l_{N, \alpha}(dx) = \int \tanh(h+ \be x)^2 l_{N, \alpha}(dx)
\eeq 
the {\it Edwards-Anderson order parameter}, indeed as claimed on \cite[p. 69]{MPV}. (The abstract version of the EA-order parameter will play a key role in our abstract models).  \\

\noindent In order to contrive tractable models we shall perform a {\it REM-approximation}: we replace the local fields by a collection of {\it independent} standard Gaussians  $h_i^\alpha \hookrightarrow g_{\alpha,i}$, where $i=1 \dots N$ and $\alpha = 1 \dots 2^N$ (the complexity of the relevant TAP-solutions being currently unknown we simply set, here and henceforth, $\Theta = 1$). Denoting by 
\beq 
r_{N, \alpha} \defeq \frac{1}{N} \sum_{i=1}^N \delta_{g_{\alpha, i}}
\eeq 
the empirical measure, we thus consider the REM-approximation of the TAP free energy
\beq \label{tapsamp_two}
f_{\text{REM-TAP}}\left( \alpha\right) = \Phi \left( \int g(x)^2\, r_{N, \alpha}(dx) \right) + \int f_1(x) r_{N, \alpha}(dx)\,.
\eeq
This leads to an approximation of the SK-model which only relies on the alleged geometrical organisation of the relevant\footnote{Again emphasizing that, at the time of writing, the meaning of "relevant" still isn't settled.}
TAP-solutions: remark in fact that for the {\it abstract overlap} it holds
\beq \bea \label{absov}
q_N(\alpha, \alpha') & \defeq \frac{1}{N} \sum_{i=1}^N \tanh(h + \be g_{\alpha, i}) \tanh(h + \be g_{\alpha', i}) \\
& \approx 
\E\left[\tanh(h + \be g_{1, 1})^2 \right] {\bf 1}_{\{\alpha = \alpha'\}}+ \E\left[ \tanh(h + \be g_{1, 1}) \right]^2  {\bf 1}_ { \{ \alpha \neq \alpha'\}}\,,
\eea \eeq
for large enough $N$, by the law of large numbers; this is indeed the "black or white dichotomy" of the REM \cite{D}, or, which is the same, the perpendicularity of TAP-solutions within a 1RSB Ansatz \cite{MPV}. \\

\noindent The above begs the following, natural questions: 
\begin{itemize}
\item[Q1.] what is the law of the overlap $q_N(\alpha, \alpha')$ under the Gibbs sampling \eqref{tapsamp_two} ? 

\item[Q2.] How does the abstract overlap transform under the {\it extensive} cavity dynamics \cite{MPV}? This amount to studying, for $\vare > 0$, the impact of an $\vare$-perturbation of the Hamiltonian \eqref{tapsamp_two}, i.e. to study the limiting Gibbs measure under transformations of the type
\beq \label{cavity_epsi_ham_intro}
f_{\text{REM-TAP}}\left( \alpha \right) \hookrightarrow  f_{\text{REM-TAP}}^{(\vare)}\left( \alpha \right) \defi f_{\text{REM-TAP}}\left(\alpha \right)+ \vare \int \log \cosh(x) \bar r_{N, \alpha}(dx),
\eeq
where 
\beq 
\bar r_{N, \alpha}(dx) \defi \frac{1}{N} \sum_{i=^1}^N \delta_{ \bar g_{\alpha, i}},
\eeq 
and with $\{\bar g_{\alpha, i}\}_{\alpha, i}$ being some {\it fresh} disorder, i.e. a random field of centered, independent Gaussians which are also independent of the {\it reservoir} $\{ g_{\alpha,i}\}_{\alpha, i}$. The limit we are interested in is, of course, the double limit $N \uparrow \infty$, followed by $\vare \downarrow 0$.  
\end{itemize}

\noindent The questions Q1 \& Q2 are addressed below, in general setting.  

\section{The REM in TAP: definition, and main results.} \label{model_new}
We start with some notation: $(S,\mathcal{S})$ denotes a Polish space and $C(S), C_b(S)$ the spaces of all real valued continuous, resp. continous bounded functions on $S$.\\
For $d\in\N$, we denote by $\psd$ the space of Borel probability measures on $S^d$, endowed with the topology of weak convergence of measures. Notice that $\psd$ is Polish itself and we can consider one of the standard metrics (e.g. Prokhorov) that makes it a complete, separable metric space. Given a measure $\nub\in \psd$ and $r>0$ we indicate with $\oball$, resp. $\cball$, the open, resp. closed ball in the metric space $\psd$ with center $\nub$ and ray $r$. Our abstract Hamiltonian, which parallels \eqref{cavity_epsi_ham_intro}, is defined through a continous functional $\Phib:\pst\to\R$ of the form 
\beq\label{form}
\Phib[\rhob]=\Phi_1\left[\rho_1\right]+\EE_{\rho_2}(f_2)
\eeq 
where $\rho_1, \rho_2\in\ps$ are the marginals of $\rhob\in\pst$ on the first, resp. second coordinate, $\Phi_1:\ps\to\R$ is a continous functional and $f_2\in C_b(S)$.
To lighten notation, we shorten $\EE_{\boldsymbol  \rho}( u) \defi \int u(\boldsymbol x) {\boldsymbol  \rho}(d \boldsymbol x)$ for the expectation of a function  $u : S^d \to \R$ w.r.t. a given $\boldsymbol \rho \in \psd$ and $\text{var}_{\boldsymbol  \rho}(u) \defi \EE_{\boldsymbol  \rho}(u^2) -{\EE}_{\boldsymbol  \rho}(u)^2$ for the variance; we also write 
$\boldsymbol  \rho \circ u^{-1}$ for the push-forward of $\boldsymbol \rho$ along a $\boldsymbol \rho$-measurable $u$. Finally we indicate with $\pi_1, \pi_2:S^2\to S$ the natural projections on the first, resp. second coordinate.\\\\
More specifically:

\renewcommand\labelitemi{$\circ$}

\begin{Def} \label{def_Ham}
$\Phib:\pst\to\R$ denotes a functional of the form \eqref{form} with
$$\Phi_1[\rho]\defeq \Phi\left( \EE_{\rho} (g^2) \right) + \EE_\rho(f_1) \qquad \forall \rho\in\ps$$
where
\begin{enumerate}
\item[H1)] $\Phi:\R\to\R$ is a twice differentiable concave function with $\Phi''(x) <0$ for every $x\in \R$;\\
\item[H2)] $g, f_1, f_2 \in C_b(S)$ with $\sup g = 1, \inf g = -1$.\\
\end{enumerate}
The random Hamiltonian of our abstract model in a configuration $\alpha$ of the configuration space $\{1,\dots,2^N\}$ for a finite volume $N\in\N$ is then defined as
\beq \label{Ham}
\Hm_N(\alpha)\defeq N\Phib\left[\Lb\right]
\eeq
where for every $\alpha \in \{1, \cdots, 2^N\}$
\beq \label{emp_meas}
\Lb \defeq \frac{1}{N} \sum\limits_{i=1}^N \delta_{(X_{\alpha,i},Y_{\alpha,i})}
\eeq
are the empirical measures associated to independent sequences of $S^2$-valued, i.i.d. random vectors $\{(X_{\alpha,i},Y_{\alpha,i})\}_{i=1}^N$ defined on a probability space $(\Omega, \mathcal{F}, \PP)$ with common distribution $\mu\smot\gamma$, for some product measure  $\mu\smot\gamma\in\mathcal{M}^+_1(S^2)$.
\end{Def}

\noindent Some words on the assumptions given in the above definition: concavity is imposed  on $\Phi$ to warrant thermodynamical stability, somewhat in line with Plekfa's convergence criterium recalled in Remark \ref{plefkarem}; the boundedness of $g$ is technically convenient, but also tailor-suited for our applications such as the SK-model. The variables $X_{\alpha,\cdot}$ correspond to the random energies associated to the configuration $\alpha$, while the $Y_{\alpha,\cdot}$ encode the disorder necessary to perturbate the Hamiltonian with an extensive cavity dynamics.\\

\noindent To shorten notation, we define

\beq \label{total_linear}
f(x,y)\defeq f_1(x)+f_2(y)
\eeq

\noindent so that the nonlinear functional $\G:\pst\to\R$ in the definition \ref{def_Ham} is
\beq
\G[\rhob] \defeq \Phi\left( \EE_{\rho_1}( g^2 )\right)+\EE_{\rhob}({f}).
\eeq

\noindent Notice that the continuity of the real valued function $\Phi$ and the boundedness of the three functions $g, f_1, f_2$, imply that the functional $\Phib$ is continous on $\pst$ by duality and continous projection.

\vskip0.5cm

\noindent We consider the partition function, free energy and Gibbs measure associated to the Hamiltonian \eqref{Ham}: for finite volume $N\in\N$, these are defined as usual as
\beq \label{free_en}
Z_N \defeq \sum\limits_{\alpha=1}^{2^N} \exp \Hm_N(\alpha), \qquad F_N \defeq \frac{1}{N}\log Z_N,
\eeq 
and for $ \alpha\in\{1,\dots,2^N\}$, 
\beq \label{gibbs_measure}
\Gm \defeq  Z_N^{-1} \exp{\Hm_N(\alpha)}.
\eeq
Finally, given two configurations $\alpha, \alpha'$, we define their abstract overlap as 

\beq \label{overlap_sez3}
\qap \defeq \frac{1}{N}\sum\limits_{i=1}^N g\left(X_{\alpha,i}\right)g\left( \, X_{\alpha',i} \right).
\eeq
Our first result concerns the limiting free energy,  and requires some notation: set
\beq\label{Kdef} 
\mathcal K \defeq \{\boldsymbol \nu\in\pst:\quad H(\boldsymbol \nu\mid\mu{\small \otimes}\gamma)\leq\log2\}\,,
\eeq
with $H(\boldsymbol \nu\mid\mu{\small \otimes}\gamma) \defeq \EE_{\boldsymbol \nu} \log \left( d\boldsymbol \nu / d(\mu{\small \otimes}\gamma) \right)$ being the usual relative entropy (see Section \ref{prop_H} for some relevant properties).

\begin{teor} (Boltzmann-Gibbs principle). \label{teo1}
The infinite volume limit of the free energy \eqref{free_en} exists $\PP$-almost surely, is non-random, and given by
\beq \label{bgvp_rem_tap}
\lim\limits_{N\to\infty} F_N = \sup_{\boldsymbol \nu\in \mathcal K} \,\, \G[\boldsymbol \nu]-H(\nub\mid\mu{\small \otimes}\gamma)+\log 2\,.
\eeq
\end{teor}
A complete solution of our abstract models thus requires a discussion of the Boltzmann-Gibbs variational principle \eqref{bgvp_rem_tap}. This will be achieved by relating it to a simpler, Parisi-like variational principle in finite dimensions.  To see how this comes about, we recall from \cite{BK1} that for all $f\in \mathcal{L}\left( S^2, \mu\smot\gamma\right)$ with
\beq\label{finite_MGF}
\mathcal{L}\left( S^2, \mu\smot\gamma\right) \defeq \left\{ u\in C\left(S^2\right): \quad \EE_{\mu\smot\gamma}\left(e^{\lambda u}\right) < \infty \,\,\, \forall \lambda \in \R\right\},
\eeq
almost surely it holds
\beq
\bea \label{sol_remcav}
\lim\limits_{N\to\infty} \frac{1}{N}\log \sum\limits_{\alpha=1}^{2^N} \exp{N\EE_{\Lb}(f)} & = \sup_{\boldsymbol \nu\in K} \left\{ \EE_{\mu{\small \otimes}\gamma}(f)- H(\boldsymbol \nu\,\mid\mu{\small \otimes}\gamma) \right\}\\
& = \inf_{0\leq m \leq 1} \left\{ \frac{1}{m} \log \EE_{\mu{\small \otimes}\gamma}(\exp{\,mf})+\frac{\log2}{m} \right\}. 
\eea \eeq

The first equality is the Boltzmann-Gibbs principle given in Theorem \ref{teo1} for a Hamiltonian of the form \eqref{Ham} without nonlinear term. The second equality, in full agreement with the Parisi theory \cite{MPV}, establishes a duality between the Gibbs principle and a finite-dimensional minimization problem. We shall call the target function in the minimization problem a {\it Parisi function}; its (unique, cfr. \cite{BK1}) minimizer $\bam$ on $[0,1]$ gives the limiting Gibbs measure as the one whose Radon-Nikodym derivative with respect to $\mu{\small \otimes}\gamma$ is given by the Boltzmann factor $\exp{\bam f}\equiv\exp{\bam (f_1\circ\pi_1+f_2\circ\pi_2)}$.\\
Our second main result provides the analogous duality principle for the nonlinear Hamiltonian \eqref{Ham}. Specifically, defining for every $(q,m)\in[0,1]^2$
\beq \bea
Z^{q,m} & \defeq \EE_{\mu{\small \otimes}\gamma}\left( \exp{m\left[\Phi'(q)g^2\circ\pi_1+f\right]}\right)\\
& = \int \exp{m\left[ f_1(x)+\Phi'(q)g^2(x)\right]}\, \mu(dx) \cdot \int \exp{m f_2(y)}\, \gamma(dy)
\eea \eeq
and 
\beq \label{Parisi}
\Pa(q,m) \defeq \Phi(q)-q\Phi'(q)+\frac{1}{m}\left( \log Z^{q,m} + \log2\right),\,
\eeq
the following holds.

\begin{teor} (Parisi principle). \label{teo2} It holds: 
$$\lim\limits_{N\to\infty} F_N = \inf_{(q,m)\in[0,1]^2} \Pa(q,m)+\log2\,,$$
$\PP$-almost surely. 
\end{teor}

When comparing the \emph{Parisi function} \eqref{Parisi} with its counterpart \eqref{sol_remcav} for the linear models one observes, in particular, the appearance of the term $\Phi(q)- q \Phi'(q)$. 
We will see in the course of the proof that such corrections,  which are constituent parts of the Parisi free energy for the SK-model,  play the role of \emph{Lagrange multipliers} accounting for the in-built nonlinearities.  \\

\noindent We now present our main results concerning the Gibbs measure.  First of all,  we note that the family of \textit{generalized Gibbs measures} 
$$\left\{\gqm \right\}_{(q,m)\in[0,1]^2}  \subset \pst$$ defined through their Radon-Nikodym derivatives with respect to $\mu{\small \otimes}\gamma$:
\beq \label{radon_candid}
\frac{d\gqm}{d(\mu{\small \otimes}\gamma)}(x,y)\defeq\frac{\exp m\left[\Phi'(q)\,g^2(x)+f(x,y)\right]}{Z^{q,m}}\,,
\eeq satisfies
\beq \label{relen_gqm} 
\bea
\log Z^{q,m}&=\EE_{\gqm}\left( m\left[\Phi'(q)\, g^2\circ\pi_1+f\right]\right)-H(\gqm\mid\mu{\small \otimes}\gamma)\\
&=m\left[\Phi'(q)\EE_{\gqmu}\left(g^2\right)+\EE_{\gqm}(f)\right]-H(\gqm\mid\mu{\small \otimes}\gamma)\,
\eea
\eeq
being $\gqmu\in\ps$ the first marginal of $\gqm$.
Moreover, one can easily compute the partial derivatives of $\Pa$ to see that
\beq \label{der} 
\bea
& \partial_m \Pa(q,m) = \frac{1}{m^2}\left[ H( \gqm \mid \mu{\small \otimes}\gamma) - \log2 \right], \\
& \partial_q \Pa(q,m) = \Phi''(q)\left[ \EE_{\gqmu}\left(g^2\right)-q\right].
\eea
\eeq
Through equations \eqref{relen_gqm}, \eqref{der} we will relate the Boltzmann-Gibbs principle \eqref{bgvp_rem_tap} with the Parisi function \eqref{Parisi}, showing that the solution of the latter is given by $\gqmb$ where $(\bar q, \bar m)$ minimizes $\Pa:[0,1]^2\to\R$.

It is furthermore well-known that the Poisson-Dirichlet law for the {\it pure states} appears naturally as the weak limit of the Gibbs measure associated to a classical REM in low temperature.  Our third main result, Theorem \ref{teo3} below, aligns with this alleged universality.

In order to formulate the statement,  we  point out that for our abstract models, {\it low temperature} corresponds to the situation where the parameter $\bar m$ which achieves the minimum in Parisi principle is such that $\bar m<1$. As will become clear below,  if the minimum point $(\bar q, \bar m)\in[0,1]^2$ of a Parisi function \eqref{Parisi} is such that $\bar m<1$,  then the measure 
$$\nuop \equiv  \boldsymbol \nu^{\bar q, \bar m} \in\ps$$ 
is the optimal measure for the Boltzmann-Gibbs principle and satisfies 
\beq \label{conditions}
\begin{cases} 
H(\nuop\mid\mu\smot\gamma) = \bar m\left[ \Phi'( \bar q)\bar q + \EE_{\nuop}(f)\right] -\log Z^{\bar q, \bar m}=\log2, \\
\EE_{\bar \nu_1}( g^2 )= \bar q,
\end{cases} \eeq
where $\bar \nu_1 \equiv \nu_1^{\bar q, \bar m}$ is the first marginal of $\nubb$; in particular,  in low temperature the side constraint on the relative entropy is saturated. 
\begin{teor} (Onset of Derrida-Ruelle cascades). \label{teo3}
Assume that at least one of the measures ${\nubb}\circ(g\circ\pi_1-\bar q)^{-1}$ or ${\nubb}\circ(f-\EE_{\nubb}(f))^{-1}$ has a density w.r.t. the Lebesgue measure on $\R$. Then, for a system in low temperature, the point process $\left\{ \mathcal{G}_N(\alpha) \right\}_{\alpha \leq 2^N}$ associated to the Gibbs measure \eqref{gibbs_measure} converges weakly as $N\to\infty$ to a Poisson-Dirichlet point process with parameter $\bam$.
\end{teor}
Let us furthermore denote by $$\langle O_N({\alpha,\alpha'}) \rangle_N \defeq Z_N^{-2}\sum_{\alpha, \alpha'} O_N(\alpha,\alpha') \Gm \Gmp$$ the average with respect to the {\it replicated Gibbs measure} of a quantity $O_N(\alpha,\alpha')$ depending on two configurations $\alpha, \alpha'$; the following then holds for the limiting law of the abstract overlap \eqref{overlap_sez3}.

\begin{prop} (Overlap concentration). \label{overlap_law}
Let 
\beq \label{two_over}
\bar q = \int g(x)^2 {\bar \nu_1}(dx),  \qquad \bar q_0 \defeq \left[ \int g(x) {\bar \nu _1}(dx) \right]^2\,.
 \eeq
Then, under the assumptions of Theorem \ref{teo3}, it holds
\beq \label{overlap_limit_eq}
\lim \limits_{N\to\infty} \repav{\delta_{\alpha=\alpha'}(\qap- \bar q)^2 }= 0,
\eeq
\beq \label{overlap_limit_neq}
\lim\limits_{N\to\infty} \repav{ \delta_{\alpha\neq\alpha'}\left( \qap - \bar q_0 \right)^2 } = 0.
\eeq
\end{prop}

We thus see that our abstract models correctly recover {\it all} of the main features of the Parisi theory under a 1RSB approximation,  in a generic setting.  For the readers convenience,  we shall conclude by briefly dwelling on the upshot of our analysis for the concrete case of the SK-model.

\subsection{SK vs.  REM-TAP}\label{link_parisi}  We recall that the 1RSB Ansatz \cite{MPV} for the limiting free energy of the SK model \eqref{SK_Ham} reads 
\beq \label{f1RSBSK} \bea 
& f_{\text{1RSB-SK}}(\beta, h)  \defeq \inf_{0 \leq q_0, \leq q_1 \leq 1,  m_1 \in [0,1]} \Bigg\{ \frac{\beta^2}{4}  \left[ (1-m_1)q_1^2 + m_1 q_0^2 - 2q_1 + 1 \right] + \log2 + \\
 & \qquad \qquad + \frac{1}{m_1} \int \log \left\{ \int \exp m_1 \log \cosh \left[ \arth \right] \varphi(dv) \right\} \, \varphi(du) \Bigg\} ,
\eea \eeq
where $\varphi$ is the standard Gaussian measure on $\R$.  As a matter of fact,   this is a {\it degenerate} 2RSB-formula:  indeed, the law of the pure states which hides behind \eqref{f1RSBSK}  is that of a superposition of two Derrida-Ruelle processes \cite{R} with parameters $0 \leq m_0 \leq m_1 \leq1$,  with the first one eventually absorbed through the limiting procedure $m_0 \downarrow 0$: this operation gives rise to the aforementioned degeneracy,  to the "common trunk" captured by the second Gaussian $du$-integral,  and stands behind the {\it  third} order parameter $q_0$.  For better comparison,  one may in fact simply set $q_0=0$ in the above formula.  Due to this complication,  comparison with our 1RSB-setting shall be taken {\it cum grano}.  Notwithstanding,  we do get a number of important insights: for the Hamiltonian $H_{N}(\alpha)\equiv N f_{\text{REM-TAP}}\left( \alpha \right)$, and recalling {SK1-3)} above,  the REM-TAP approximation leads in fact to a free energy  
\beq \bea \label{reta}
& f_{\text{REM-TAP}}(\be, h) = \inf_{0 \leq q \leq 1, m \in [0,1]} \Bigg\{ 
 \frac{\be^2}{4} (q^2+1) + \log 2 + \\
& \qquad \qquad + \frac{1}{m} \log  \int \exp m \log \cosh(h+\beta x) \cdot e^{m\left[ -\frac{\beta}{2} x \tanh(h+\beta x)-\frac{\beta^2}{2}q\tanh^2(h+\beta x)\right]}\, \varphi(dx) 
\Bigg\}
\eea \eeq
(The term $\exp{m\left[ -\frac{\beta}{2} x \tanh(h+\beta x)-\frac{\beta^2}{2}q\tanh^2(h+\beta x)\right]}$ in \eqref{reta} plays a role in the cavity dynamics only,  and is completely irrelevant when it comes to the free energy: for the sake of the current discussion,  such factor is an artefact which can also be immediately removed : simply set $f_1 \equiv 0$).  The appearance of the {\it quadratic} term in the first line on the r.h.s.  above is central to this work: although a constituent part of \eqref{f1RSBSK},  such terms cannot be explained by the linear models studied in \cite{BK1,  BK2}.  As already mentioned,  these corrections turn out to be Lagrange multipliers accounting for the non-linearities induced by the true REM-models hiding "within" the TAP-free energies. \\

\noindent Let us also spend a few words on the cavity dynamics for the SK-model,  i.e.  for the Hamiltonian given by the perturbed REM-TAP $H_{N;\varepsilon}(\alpha)\defeq N f_{\text{REM-TAP}}^{(\vare)}\left( \alpha \right)$, and 
with $f_2(y)\equiv f_2(\varepsilon;y)\equiv \varepsilon \log\cosh(y)$.  Let us denote by $(\bar q_\varepsilon, \bar m_\varepsilon)\in[0,1]^2$ the minimum of the corresponding Parisi function, and by  $\bar \nub_\varepsilon = \nub^{\bar q_\varepsilon, \bar m_\varepsilon}$ the associated generalized Gibbs measure.  We furthermore denote by $F_{N, \vare}$ the law of the overlap $q_N(\alpha, \alpha') = (1/N) \sum_{i=1}^N \tanh(h+\be g_{\alpha, i})\tanh(h+\be g_{\alpha', i})$ under the perturbed,  finite volume Gibbs measure,  and by $F_N$ that of the {\it un}-perturbed: 
\beq
x \in [0,1] \mapsto F_{N, \vare}(x) \defi \repave{{\bf 1}_{q_N(\alpha,\alpha')\leq x}}, \qquad F_{N}(x) \defi F_{N, 0}(x).
\eeq
Proposition \ref{overlap_law} would imply that in low temperature ($\bar m_\varepsilon < 1$) the law of the overlap of two relevant TAP-solutions for the SK-model is given by 
\beq \label{ovelap_quenc_cdf} 
F_\vare(x) \defeq \lim_{N\to \infty} F_{N, \vare}(x) = \begin{cases} 0 \quad & \text{if} \quad x<\bar q_{0;\varepsilon}\\
\bar m_\varepsilon  \quad & \text{if} \quad \bar q_{0;\varepsilon} \leq x <  \bar q_\varepsilon  \\
1  \quad & \text{otherwise}\end{cases};
\eeq
where 
$$\bar q_\varepsilon = \int \tanh^2(h+\beta x) \bar \nub_\varepsilon(dx,dy), \qquad \bar q_{0;\varepsilon} = \left[ \int \tanh(h+\beta x) \bar \nub_\varepsilon(dx,dy)  \right]^2.$$
A Parisi fixed point equation \cite[III.63]{MPV} encoding the stability of the hierarchical structure under the extensive cavity dynamics would then appear through the {\it imposition}  
\beq \label{requir}
F(x) \stackrel{(!)}{=} \lim_{\vare \downarrow 0} F_\vare(x),  \qquad \forall \, x \in [0,1].
\eeq
In case of a REM-approximation (1RSB),  it is easily seen that the self-consistency \eqref{requir} is, in fact,  {\it automatically} satisfied and as such a void requirement: in order to avoid trivializations one needs a much more sophisticated GREM-TAP approximation,  but this will be addressed in forthcoming works.

\section{Proofs.}

\subsection{The Boltzmann-Gibbs principle: proof of Theorem \ref{teo1}.}

\vskip0.3cm

In this subsection we prove Theorem \ref{teo1}, i.e. the validity of the following Boltzmann-Gibbs principle
\beq \label{BGVP}
\bea
\lim\limits_{N\to\infty} F_N = \lim\limits_{N\to\infty} \frac{1}{N}\log Z_N& =\sup_{\nub\in \mathcal{K}}\, \Omega[\nub]+ \log2  \quad \PP-\text{a.s.}, \\
\eea \eeq
where to lighten notation we shortened
\beq \label{omega_func}
\bea 
\Omega[\rhob] & \defeq \Phib[\rhob] - H(\rhob\mid\mu\smot\gamma)\\
& = \Phi\left( \EE_{\rhob_1}( g^2 )\right)+\EE_{\rhob}({f}) - H(\rhob\mid\mu\smot\gamma)
\eea
\eeq
being $\Phib:\pst\to\R$ the functional defined in \ref{def_Ham}.

\begin{rem} \label{exist_uniq}
	It is a well known fact that $\nub\in\pst\to H(\nub\mid\mu\smot\gamma)$ is lower semicontinous, strictly convex and with compact sub-levels.  Therefore, for compact $\mC\subset \pst$ and an upper semicontinous $\Omega$, the generalized Bolzano-Weierstrass Theorem on metric spaces ensures the existence of a solution to 
	\beq \label{BGVP_general_C}
	\sup_{\nub\in \mC} \Omega[\nub]. \eeq
	Moreover, the assumption H1) on $\Phi$, together with the fact that $g\in[-1,1]$, suffices to get strict concavity of the functional $\Omega:\pst \to \R$. This implies that, for a convex  $\mC\subseteq \pst$, the existence of a solution to a problem of the form \eqref{BGVP_general_C} implies its uniqueness. In other words: if $\mC$ is compact,  there exists a solution to \eqref{BGVP_general_C}; if $\mC$ is also convex, such solution is furthemore unique: since the set $\mK$, i.e. the sub-level of the relative entropy appearing in the Boltzmann-Gibbs principle \eqref{BGVP}, is convex and compact,  the variational principle \eqref{BGVP} admits a unique solution.  (Recall also that, in a complete metric space like $\pst$ a set is sequentially compact if and only if it is compact, every compact set is closed and every closed subset of a compact set is compact).
\end{rem}
\vskip0.3cm

We will proceed via a second moment method on the counting variable
$$\cv{E}=\#\{\alpha:\quad \Lb \in E\},$$
defined for every $E\subseteq\pst$. The proof is analogous to the one given in \cite{BK1} and relies both on the independence of the $\{\Lb\}_{\alpha=1}^{2^N}$ as well as on the continuity of $\rhob\in\pst\to\Phib[\rhob]$. We start with a technical Lemma, which will be useful also later, that provides the upper bound to the free energy - see equation \eqref{upperbound}.  Here and henceforth we denote by $B_{r} \defeq B_{\bar \nub, r} = \{\nub \in \pst:\quad d(\nub,\bar \nub)<r\}$ the open ball centered in $\bar \nub$, and with radius $r>0$.

\begin{lem} \label{ub_lemma}
	Let $\bar \nub \defeq \text{argmax}_{\,\mathcal{K}\,} \Omega$.  Then,  $\PP$-almost surely : 
\begin{itemize}
\item[$\bullet$] For all $r > 0$,
	$$ \limsup\limits_{N\to \infty} \frac{1}{N}\log \sum_{\alpha: \Lb \not\in B_r} e^{N\Phib[L_{N,\alpha}]} < F, $$

\item[$\bullet$] With $F\defeq\max\limits_{\mK} \Omega[\nub] = \Omega[\bar \nub]$,  it holds: 
\beq \label{upperbound}
\limsup\limits_{N\to \infty}\frac{1}{N}\log \sum\limits_{\alpha=1}^{2^N} e^{N\Phib[\Lb]} \leq F\,.
\eeq
\end{itemize}

\end{lem}

\begin{proof}
	Let $r\geq 0$.  Shorten $B_0^c \equiv \pst$, denote by $B_r^c$ the complement of $B_r$ in $\pst$: 
	$$B^c_{r} \defeq \{\rhob \in \pst:\quad d(\rhob,\bar \nub)\geq r\})=\pst\setminus B_r,$$ 
	and set
	\beq \label{BGVPepsi}
	F_r \defeq \begin{cases} \bea
	&\sup\limits_{\mK\cap B_r^c}\Omega \quad &\text{if} \quad \mK\cap B_r^c\neq \emptyset,\\
	&-\infty \quad& \text{otherwise}. \eea \end{cases} 
	\eeq
	Notice that if $r:\,\, \mK\cap B^c_r= \emptyset$ it trivially holds $F_{r} < F_0 \equiv F$. If instead $r:\,\, \mK\cap B_r^c\neq \emptyset$ then, by compactness  (see Remark \ref{exist_uniq}), we can find $\nub_r \in \mK\cap B^c_r, \,\,\nub_r\neq\nubb$ s.t.
	$$F_r = \max_{\mK\cap B^c_r} \Omega = \Omega[\nub_r]$$
	and, as the maximal measure $\bar \nub$ of $\Omega$ on $\mK$ is unique, 
	 the strict inequality 
	$$F_r= \Omega[\nub_r] < \Omega[\bar \nub] = F$$
	holds also in this case.	
	\noindent Hence, in order to settle the Lemma, we just need to prove the following 
	\beq \label{claim1}
	\limsup\limits_{N\to \infty} \frac{1}{N}\log \sum_{\alpha: \Lb \in B^c_r} e^{N\Phib[\Lb]}\leq F_r \quad \forall r \geq 0.	\eeq
	
In order to see this,  consider a decreasing sequence of positive real numbers $\{a_n\}_{n\in\N}$ with $a_n \in (0,1)$ for all $n\in\N$ and $a_n\to 0$ as $n\to\infty$. Then fix $n\in \N$ and consider the open set $\mK_{a_n}^c\defeq \pst\setminus \mK_{a_n}$, together with its closure $\overline{\mK^c_{a_n}}$, where for every $a\geq0$ 
	$$\mK_a \defeq \{\nub\in\pst:\quad H\left( \nub\mid \mu\smot\gamma \right) \leq \log2 + a \}.$$
	  Sanov's Theorem implies  that 
	\beq  \label{Sanov_closure_interior}
	\bea
	\PP\left(\Lbu \in \overline{\mK_{a_n}^c} \right) & = \exp{[-\inf_{\rhob\in  \overline{\mK_{a_n}^c}} H(\rhob\mid \mu\smot\gamma)N+o(N)]} \quad \\
	& \leq 2^{-N} e^{-a_nN+o(N)}
	\eea(N\to\infty), \\
	\eeq
	hence
	\beq \label{N_sets_for_BC}
	\bea 
	\PP\left( \exists \alpha\in\{1,\cdots,2^N\}:\quad \Lb\in \mK_{a_n}^c\right) &  \leq 2^N\PP\left( \Lbu\in \overline{\mK_{a_n}^c}\right)\leq e^{-a_nN+o(N)}. \eea\eeq
	
	\noindent Through the Borel-Cantelli lemma, \eqref{N_sets_for_BC} implies that 
	
	\beq \label{not_out_Ka}
	\bea
	\limsup\limits_{N\to\infty} \frac{1}{N}\log \sum_{\alpha: \Lb \in B^c_r} e^{N\Phib[\Lb]}
	&= \limsup\limits_{N\to\infty} \frac{1}{N} \log\sum\limits_{\alpha:\Lb\in \mK_{a_n}\cap B^c_r} e^{N\Phib(\Lb)}\quad \PP-\text{a.s.}\eea
	\eeq
	
	If $r, n$ are such that $\mK_{a_n}\cap B^c_r=\emptyset$, the latter implies
	\beq \label{claim1empty}
	\bea
	\limsup\limits_{N\to\infty} \frac{1}{N}\log \sum_{\alpha: \Lb \in B^c_\epsilon} e^{N\Phib[\Lb]}
	&= -\infty \quad \PP-\text{a.s.}\eea
	\eeq
	
	Notice that by construction it holds  $\mK=\mK_{0}\subseteq \mK_{a_n}$, which implies that if $\mK_{a_n}\cap B^c_r$ is empty then $F_r = -\infty$, i.e. the claim \eqref{claim1} is exactly \eqref{claim1empty}. 
	
	We can therefore assume $r, n:\,\,\mK_{a_n}\cap B^c_r\neq \emptyset$ and, for an aribitrary $\delta>0$ cover the compact set $\mK_{a_n}\cap B^c_r$ through open balls  
	$$\mK_{a_n}\cap B^c_r \subset \bigcup_{\nub \in \mK_{a_n}\cap B^c_\epsilon} B_{\nub}(r_{\nub})$$ 
	such that for each $\nub\in \mK_{a_n}\cap B^c_r$ the associated ray $r_{\nub}>0$ is small enough s.t.
	\beq \label{radius_continuity_phi}
	\Phib(\rhob)-\Phib(\nub) < \delta \quad \forall \rhob \in B_{\nub}(r_{\nub}),
	\eeq
	and
	\beq \label{entropy_min_centre}
	H(\nub\mid \mu{\small \otimes}\gamma)-
	\inf_{\rhob \in \overline{B_{\nub}(r_{\nub}) }} H(\rhob\mid\mu\smot\gamma)<\frac{\delta}{2}.
	\eeq
	
	This is clearly possible by upper semicontinuity of $\Phib$ and the fact that 
	$$\inf_{\rhob \in \overline{B_{\nub}(r) }} H(\rhob\mid\mu\smot\gamma)\to H(\nub\mid \mu{\small \otimes}\gamma)\quad (r\to 0).$$
	
	By compactness, we can extract a finite sub-cover: $\exists M<\infty, \{\nub_i\}_{i=1}^M\subset \mK_{a_n}\cap B^c_r $ s.t. $\mK_{a_n}\cap B^c_r \subset \bigcup_{i=1}^M B_{\nub_i}(r_i)$, where we shortened 
	$$r_i\defeq r_{{\nub}_i}, \,\, B_i\defeq B_{\nub_i}(r_i).$$
	
	Therefore, almost surely, it holds
	
	\beq \label{bound_count_var}
	\bea
	\limsup\limits_{N\to\infty} \frac{1}{N}\log \sum_{\alpha: \Lb \in  \mK_{a_n} \cap  B^c_r} e^{N\Phib[\Lb]} & \leq \limsup\limits_{N\to\infty} \frac{1}{N}\log \sum\limits_{i=1}^M\sum\limits_{\alpha:\Lb\in B_i} \exp{N\Phib(\Lb)}\\
	& \leq \delta + \limsup\limits_{N\to\infty} \frac{1}{N}\log \sum\limits_{i=1}^M \exp{N\Phib(\nub_i)}\cv{\,\overline{B}_i\,}
	\eea
	\eeq
	(recall that $\cv{\,\overline{B}_i\,}$ is the variable that counts the $\alpha-$s for which $\Lb\in B_i$).\\
	
	\noindent Markov's inequality implies that for each $i=1,\cdots, M$
	\beq \label{probub} \bea
	\PP\left(\cv{\overline{B}_i} \geq  e^{N[\log2-H(\nub_i\mid\mu{\small \otimes}\gamma)+\delta]} \right) & \leq 2^{-N}\E \, \cv{\overline{B}_i}\exp{N[H(\nub_i\mid\mu{\small \otimes}\gamma)-\delta]} \\
	\eea
	\eeq
	with
	\beq \label{expub}
	\bea
	\E\left(\cv{\overline{B}_i} \right)&=2^N\PP\left( \Lbu\in \overline{B}_i\right)=2^N\exp{[-\inf_{\rhob\in\overline{B}_i}H\left(\rhob\mid\mu\smot\gamma \right)N+o(N)]} \quad (N\to\infty),
	\eea
	\eeq
	the last step by Sanov's theorem. Plugging \eqref{expub} into \eqref{probub}, and using \eqref{entropy_min_centre}, we find that for each $i\in\{1,\cdots,M\}$
	
	$$\PP\left(\cv{\overline{B}_i} < e^{N[\log2-H(\nub_i\mid\mu{\small \otimes}\gamma)+\delta]} \right) \geq 1- \exp\left[{-\frac{\delta}{2}N+o(N)}\right] \quad (N\to\infty).$$
	
This, together with Borel-Cantelli, implies
	$$\bea
	& \limsup\limits_{N\to\infty} \frac{1}{N}\log \sum_{\alpha: \Lb\in  \mK_{a_n}\cap \in B^c_r} e^{N\Phib[\Lb]} \\
& \hspace{3cm}  \leq \limsup\limits_{N\to\infty} \frac{1}{N}\log \sum\limits_{i=1}^M \exp{N(\Phib(\nub_i)-H(\nub_i\mid\mu{\small \otimes}\gamma))}+\log2+2\delta\\
	&  \hspace{3cm} \leq \sup_{\nub \in B^c_\epsilon \cap \mK_{a_n}} \Phib(\nub)-H(\nub\mid\mu{\small \otimes}\gamma)+\log2+2\delta.
	\eea$$
	
	Being $\delta$ arbitrary, the latter, together with \eqref{not_out_Ka} give
	
	$$\bea
	\limsup\limits_{N\to\infty} \frac{1}{N}\log \sum_{\alpha: \Lb \in B^c_r} e^{N\Phib[\Lb]} & \leq \sup_{\nub \in \mK_{a_n}\cap B^c_r} \Omega[\nub]+\log2. \\
	\eea$$
	
	Notice that, again by compactness, we can select a measure $\bar \nub_n\in \mK_{a_n}\cap B^c_r$ s.t. 
	$$\max\limits_{\mK_{a_n}\cap B^c_r} \Omega = \Omega[\bar \nub_n],$$ which implies that, for every $n\in\N$,
	
	$$\bea
	\limsup\limits_{N\to\infty} \frac{1}{N}\log \sum_{\alpha: \Lb \in B^c_r} e^{N\Phib[\Lb]} &
	& \leq \Omega[\bar \nub_n]+\log2.
	\eea$$
	
	By construction it holds $B^c_r \cap \mK_{a_n} \subseteq B^c_r \cap \mK_{1}\,\, \forall n\in\N$, specifically the whole sequence $\{\bar \nub_n\}_{n\in\N}$ lies in the compact set $B^c_r \cap \mK_{1}$ and we can extract a convergent subsequence: $\{\snk\}_{k\in\N}$ s.t. $\snk \to \xib$ weakly as $k\to\infty$ for some $\xib\in \mK_1\cap B^c_r$.  We get
	\beq \label{last_bound}
	\bea
	\limsup\limits_{N\to\infty}\frac{1}{N}\log \sum_{\alpha: \Lb \in B^c_r} e^{N\Phib[\Lb]}& \leq 
	\limsup_{k\to\infty} \, \Omega[\snk]+\log2\\
	& \leq \Omega[\xib]+\log2,
	\eea \eeq
	the last step by upper semicontinuity of $\Omega$.  It is also immediately checked that 
	\beq \label{limit_in_K}
	\xib\in \mK\cap B^c_r \equiv  \mK_0\cap B^c_r,
	\eeq which implies the claim \eqref{claim1} straightforwardly through \eqref{last_bound}.\\

\noindent	The validity of \eqref{limit_in_K} is a consequence of the fact that  for every $\delta>0$ there exists $k_\delta>0$ s.t. $\snk\in \mK_\delta\cap B^c_r$ for all $k\geq k_\delta$ and that by compactness of the latter it must hold $\xib\in \mK_\delta\cap B_r^c$. Specifically:
	$$\bea
	\xib \in \bigcap\limits_{\delta>0} K_\delta \cap B^c_r &= \bigcap\limits_{\delta>0} \left\{ \nub: H(\nub\mid\mu\smot\gamma)\leq \log2+\delta\right\} \cap B^c_r\\
	& = \left\{ \nub: H(\nub\mid\mu\smot\gamma)\leq \log2\right\} \cap B^c_r = \mK\cap B^c_r,\eea$$
	and the Lemma follows.
\end{proof}

\vskip1cm

The proof of the lower bound requires a refinement of our moments' analysis. A key ingredient is the strong concentration of the $\cv{\,V}$-random variable around its mean:

\begin{lem}{(Variance estimate)} \label{2nd_moment}
For every $\nub\in \mK$ and any open neigborhood $U$ of $\nub$, there exists $\oball\subset U$ and $\delta>0$ such that for large enough $N$
\beq \label{2nd_mom}
\mathbb{V}ar \cv{\oball} \leq e^{-N\delta} (\E\,\cv{\oball})^2.
\eeq
\end{lem}
\begin{proof}
For each $\oball$,  and by independence, the second moment of $\cv{\oball}$ satisfies
\beq \label{step1_2m}
\begin{aligned}
\E\left[\,\cvsquared{\oball}\right]  & = \sum\limits_{\alpha} \mathbb{P}\left(\Lb\in \oball\right) + \sum\limits_{\alpha\neq\alpha'}\mathbb{P}\left(\Lb\in \oball \right)\mathbb{P}\left(\boldsymbol{L}_{N,\alpha'}\in \oball \right) \\
& \leq2^N\mathbb{P}\left(\boldsymbol{L}_{N,1}\in \cball \,\right) + \left[\E\, \cv{\oball} \right]^2,\eea
\eeq
and thus 
\beq \label{step1.2_2m}
\mathbb{V}ar \cv{\oball} \leq 2^N\mathbb{P}\left(\boldsymbol L_{N,1}\in \cball \,\right).
\eeq
Notice that the statement of the Lemma is trivial if $\nub=\mu{\small \otimes}\gamma$, therefore we assume $\nub\neq \mu{\small \otimes}\gamma$; this implies that there exists $\oball\subset U$ and $\eta>0$ such that $\mu{\small \otimes}\gamma\notin\cball$ and 
$$
H(\oball)=H(\cball)=H(\nub\mid\mu{\small \otimes}\gamma)-\eta.
$$
Together with Sanov's theorem, the latter implies that for $N$ large
\beq \label{step2_2m}
\bea
2^N\mathbb{P}\left( \boldsymbol L_{N,1}\in \cball \right) & \leq\exp N\left[\log2-H(\nub\mid\mu{\small \otimes}\gamma)+\frac{7}{6}\eta\right]\\
& \leq e^{-\frac{\eta}{2}N}\exp 2 N\left[ \log2-H(\nub\mid\mu{\small \otimes}\gamma)+\frac{5}{6}\eta \right]\\
& \leq e^{-\frac{\eta}{2}N} \left[2^N\mathbb{P}\left( \boldsymbol L_{N,1}\in \oball \right)\right]^2=e^{-\frac{\eta}{2}N}\left[\E\, \cv{\oball}\right]^2 \\
\eea
\eeq
where in the second line we used that, as $\nub\in K$, $\log 2-H(\nub\mid\mu{\small \otimes}\gamma)\geq 0$.  The Lemma now follows from \eqref{step1.2_2m} with $\delta\equiv\frac{\eta}{2}$.

\end{proof}

For $\nub\in \mK$, $\delta>0$ let $U$ be an open neigborhood of $\nub$ such that 
$$ \Phib(\nub)-\Phi(\rhob) \leq \delta\qquad \forall\,\rhob\in U,$$
and let $\cball\subset U$ be the ball prescribed by Lemma \ref{2nd_moment}. Then 
$$ \begin{aligned}
\liminf\limits_{N\to\infty} \mathcal{F}_N  &\geq \liminf\limits_{N\to\infty} \frac{1}{N}\log\sum\limits_{\alpha:\, \Lb\in \oball} \exp{N\Phib(\Lb)}\\
& \geq \Phib(\nub) + \liminf\limits_{N\to\infty} \frac{1}{N}\log\cv{\oball} -\delta.
\eea
$$
Moreover, Lemma \eqref{2nd_moment} together with Chebyshev's inequality and Borel-Cantelli,  immediately implies that 
$$\lim\limits_{N\to\infty} \frac{1}{N}\log\cv{\oball} = \lim\limits_{N\to\infty} \frac{1}{N}\log\E\,\cv{\oball} \quad\quad \mathbb{P}\text{-a.s.}$$
By Sanov's theorem, we thus get 
$$ \liminf\limits_{N\to\infty} \frac{1}{N}\log\cv{\oball} \geq -\inf_{\rhob\in\oball}H(\rhob\mid\mu\smot\gamma)+\log2 \quad\quad \mathbb{P}\text{-a.s.}$$
All in all,
$$\bea 
\liminf\limits_{N\to\infty} \mathcal{F}_N & \geq \Phib(\nub) -\inf_{\rhob\in\oball}H(\rhob\mid\mu\smot\gamma)+\log2-\epsilon\\ 
& \geq  \Phib(\nub) - H(\nub\mid\mu{\small \otimes}\gamma)+\log2-\delta.\eea$$
Since $\delta$ is arbitrary,  by taking the supremum over $\nub\in \mK$, we get the lower bound
$$\liminf\limits_{N\to\infty} \mathcal{F}_N \geq \sup\limits_{\nub\in \mK} \Phib(\nub) - H(\nub\mid\mu{\small \otimes}\gamma)+\log2,$$
and Theorem \ref{teo1} follows.

\subsection{The Parisi principle: proof of Theorem \ref{teo2}.}

Our proof of the Parisi principle crucially relies on some classic properties of the relative entropy functional $\rhob\in\pst\to H(\rhob\mid\mu{\small \otimes}\gamma)\in [0,\infty]$ which are recalled in the Appendix \ref{prop_H}. We start with a couple of technical and straightforward results.

\begin{lem}\label{signs_derivatives}
Let $(\bq,\bam)$ be a minimum point of the Parisi function $\Pa$, defined in \eqref{Parisi}, on $[0,1]^2$. Then 
\beq \label{signs_dm}
\bea
\bam\in(0,1)\quad &\Rightarrow \quad H(\gqmb\mid\mu{\small \otimes}\gamma) = \log2, \\
\bam=1\quad &\Rightarrow \quad H(\gqmb\mid\mu{\small \otimes}\gamma) \leq \log2, \\
\eea
\eeq
\beq\label{signs_dq}
\EE_{\gqmub}\left(g^2\right) = \bq
\eeq
being $\gqmub\in\ps$ the marginal of $\gqmb\in\pst$ on the first coordinate.\\
Moreover, it holds 
\beq \label{candidate_eq_par}
\G[\gqmb]-H(\gqmb\mid\mu{\small \otimes}\gamma)=\Pa(\bq,\bam)+\log2.
\eeq
\end{lem}
\begin{proof}

As a lower semicontinuous function on a compact set, $\Pa$ must attain a minimum on $[0,1]^2$, which we denote by $(\bq,\bam)\in[0,1]^2$.
Since $\Pa(q,m)\to +\infty$ for $m\to 0^+$, $\bam$ must lie in $(0,1]$: this implies that $\partial_m \Pa(\bq,\bam)\leq 0$ with equality if $\bam\in(0,1)$. Through this and the expression for $\partial_m \Pa$ in \eqref{der} we get that the implications \eqref{signs_dm} hold true.\\
Now, assume {\it ad absurdum} 
\beq\label{absurd}
\EE_{\gqmu}\left(g^2\right)\neq\bq.
\eeq
Since $\Phi''<0$ the assumption \eqref{absurd} implies $\bq \in \{0,1\}$, as otherwise one would get from the equation for $\partial_q P$ in \eqref{der} that $\partial_q \Pa(\bq,\bam)\neq 0$ which is impossible if $\bq\in (0,1)$. As $0$ is a left border value in order for it to be a component of a minumum point it should hold $\partial_q \Pa(\bq,\bam) \geq 0 $. But this together with \eqref{absurd}, \eqref{der} and $\Phi''<0$ would imply $\EE_{\gqm_1}\left(g^2 \right) < \bq = 0$, which contradicts the assumption $-1\leq g \leq 1$. Using that $\sup g^2 \leq 1$, the case $\bq = 1$ brings to an analogous contradiction. It thus follows that also \eqref{signs_dq} holds true.\\

Moreover, as \eqref{relen_gqm} is valid for all couples $(q,m)$, we find that $\Pa$ can be rewritten as 
\beq\label{Parisi_alt}
\bea
\Pa(q,m)=\Phi(q)-q\Phi'(q) &+ \EE_{\gqm}\left(f+\Phi'(q)g^2\circ\pi_1 \right)\, +\\ & \, +\, \frac{1}{m}\left( \log2-H(\gqm\mid\mu{\small \otimes}\gamma)\right)-\log2.
\eea
\eeq
Using \eqref{signs_dq} we get
\beq 
\bea
\Pa(\bq,\bam)=&\,\Phi(\bq)-\bq\Phi'(\bq)+\EE_{\gqmb}\left( f \right)+\Phi'(\bq)\bq+\frac{1}{\bam}\left( \log2-H(\gqmb\mid\mu{\small \otimes}\gamma)\right) - \log2\\
=&\,\Phi\left(\EE_{\gqmb}\left( g^2\circ\pi_1 \right)\right)+\EE_{\gqmb}\left( f \right)+\frac{1}{ \bam}\left( \log2-H(\gqmb\mid\mu{\small \otimes}\gamma)\right) - \log2.
\eea
\eeq
From the latter and \eqref{signs_dm} the equation \eqref{candidate_eq_par} follows straighforwardly. This concludes the Proof of the Lemma.
\end{proof}

\begin{prop}\label{prop2} There exists a minimum point $(\bq, \bam)$ of $\Pa$ on $[0,1]^2$ such that the unique optimal measure for the Boltzmann-Gibbs principle \eqref{bgvp_rem_tap} is  $\gqmb$; i.e. 
$$\sup_{\nu\in \mK}\, \G(\nub)-H(\nub\mid\mu{\small \otimes}\gamma) = \G(\gqmb)-H(\gqmb\mid\mu{\small \otimes}\gamma)$$
where $\gqmb$ is the generalized Gibbs measure with Radon-Nikodym derivative w.r.t. $\mu{\small \otimes}\gamma$ given by \eqref{radon_candid} for $(q,m)=(\bq,\bam)$.
\end{prop}
\begin{proof}
Consider $\bar{\nub}\in \mK$ solution of the Boltzmann-Gibbs principle \eqref{bgvp_rem_tap}.
Then, let $(\bq,\bam)$ be a minimum point of $\Pa$ on $[0,1]^2$. Lemma \eqref{signs_derivatives}, specifically \eqref{signs_dm}, ensures that $\gqmb\in K$; this amounts to say that $\gqmb$ is a viable candidate to solve the Boltzmann-Gibbs principle.\\

We now use Proposition \ref{ub_H_Gu} from the Appendix \ref{prop_H}, recalling the variational expression for relative entropy functionals, to show that

\beq \label{upper_candid}
\Phib(\nub)-H(\nub\mid\mu{\small \otimes}\gamma)\leq \Phib(\gqmb)-H(\gqmb\mid\mu{\small \otimes}\gamma)\quad \forall \nub\in K,
\eeq

which proves the Proposition. Specifically, as $g, f_1, f_2\in C_b(S)$, we can apply \ref{ub_H_Gu} to $\bam(\Phi'(\bq) g^2\circ\pi_1+f)\in C_b(S^2)$ to see that
\begin{equation}\label{ub_H_qm}
\bea
&H(\nub\mid\mu{\small \otimes}\gamma) \geq  H(\gqmb\mid\mu{\small \otimes}\gamma) + \bam\int \left( f +\Phi'(\bq)g^2\circ\pi_1\right)d(\nub-\gqmb)
\eea
\eeq
for all $\nub\in \pst$. Specifically
\begin{equation}\label{ub_H_qm}
\bea
\EE_{\nub}(f) -H(\nub\mid\mu{\small \otimes}\gamma) &\leq\EE_{\nub}(f) -H(\gqmb\mid\mu{\small \otimes}\gamma) -\bam\int \left( f +\Phi'(\bq)g^2\circ\pi_1\right)d(\nub-\gqmb)\\
&\bea\, =\EE_{\gqmb}&(f)  -H(\gqmb\mid\mu{\small \otimes}\gamma) -\int \Phi'(\bq)g^2\circ\pi_1d(\nub-\gqmb) \,+\\&+(1-\bam)\int \left( f +\Phi'(\bq)g^2\circ\pi_1\right)d(\nub-\gqmb)\eea\\
&\bea = -\Phi'(\bq)&\left(\,\EE_{\nub}\left( g\circ\pi_1 \right)-\bar q\, \right) + \EE_{\gqmb}(f)-H(\gqmb\mid\mu{\small \otimes}\gamma)\,+\\&+(1-\bam)\int \left( f +\Phi'(\bq)g^2\circ\pi_1\right)d(\nub-\gqmb)\eea\\
\eea
\eeq\\
where in the last line we used that, as shown in Lemma \ref{signs_derivatives}, for a minimum point $(\bar q, \bar m)$ of $P$ it holds  
\beq\label{q_eq}
\bq=\EE_{\gqmb}\left(g^2\circ\pi_1\right).
\eeq\\
Since $\Phi$ is concave and differentiable, it  satisfies
$\Phi(x)- \Phi'(y)(x-y)\leq \Phi(y)$  for all $x,y\in \R$. Through \eqref{ub_H_qm}, and again \eqref{q_eq}, this implies
\begin{equation} \label{fifty}
\bea \Phib(\nub)-H(\nub\mid\mu{\small \otimes}\gamma)&=\Phi\left( \EE_{\nub}\left( g\circ\pi_1 \right) \right) +  \EE_{\nub}\left( f \right) -H(\nub\mid\mu{\small \otimes}\gamma).\\
& \leq   \Phib(\gqmb)-H(\gqmb\mid\mu{\small \otimes}\gamma)+(1-\bam)\int \left( f +\Phi'(\bq)g^2\circ\pi_1\right)d(\nub-\gqmb).
\eea\eeq\\
From the latter we see that if $\bam=1$ the claim \eqref{upper_candid} follows immediately.
If instead $\bam\in(0,1)$ then $H(\gqmb\mid\mu\smot\gamma)=\log2$ and Proposition \ref{ub_H_Gu} implies that
$$\bam\int\left( f + \Phi'(\bq)g^2\circ\pi_1\right) d(\nub-\gqmb) \leq H(\nub\mid\mu\smot\gamma)-\log2\leq 0$$
for all $\nub\in K$. 
Since both $\bam$ and $1-\bam$ are non-negative this implies
$$  (1-\bam)\int \left( f +\Phi'(\bq)g^2\circ\pi_1\right)d(\nub-\gqmb)\leq 0$$
which used on \eqref{fifty} gives \eqref{upper_candid}, which is therefore now proved also for $\bam\in(0,1)$.\\
All in all, we have shown that if $(\bq,\bam)$ is a minimum point of $\Pa$ on $[0,1]^2$ then the unique extremal measure $\bar{\nub}$ of the Boltzmann-Gibbs principle \eqref{bgvp_rem_tap} must be one of the $\gqmb$, namely the thesis of the Proposition is settled. 
\end{proof}

\subsection{The  limiting Gibbs measure: proof of Theorem \ref{teo3}.}

Here and henceforth, we will denote by ${\nubb}\in\pst$ the measure solving the Boltzmann-Gibbs variational principle for a system in low temperature; i.e.
$$ d{\nubb}(x,y)\defeq\frac{\exp \bar m\left[\Phi'(\bar q)\,g^2(x)+f(x,y)\right]}{Z^{\bar q, \bar m}}\, d\mu(x)d\gamma(y)\,,$$
where $(\bar q, \bar m) \defeq \text{argmin}_{[0,1]^2} \Pa$ and $\bar m < 1$.
Some notation: let $\{(Z_i,W_i)\}_{i\leq N}$ be i.i.d. $S^2$-valued random vectors with common distribution ${\nubb}$ and defined on a probability space $(\Omega', \mathcal F', \EE)$. Setting $\boldsymbol Z \defeq (Z_1,W_1, \cdots, Z_N,W_N)\in S^{2N}$ define 
$$X_N({\bf Z})\defeq \frac{1}{\sqrt{N}} \sum\limits_{i=1}^N  \left(g^2(Z_i) - \bar q\right), \quad\quad Y_N({\bf Z})\defeq \frac{1}{\sqrt{N}} \sum\limits_{i=1}^N  \left( f(Z_i,W_i) - \EE_{{\nubb}}(f)\right)$$
and consider the real valued random vectors
\beq \label{vectors}
\vecb{X_N}{Y_N} \equiv \vecb{X_N(\boldsymbol Z)}{Y_N (\boldsymbol Z) },
\eeq

together with their covariance matrix, say $\Sigma$, which we assume to be invertible.\\

Set
\beq
C_{\bar q, \bar m} \defeq -\bar m\Phi''(\bar q)+\left(1,\,-\Phi'(\bar q) \right) \cdot \Sigma^{-1} \cdot \left(1,\,-\Phi'(\bar q) \right)^{\top},
\eeq 

where $\left(1,\,-\Phi'(\bar q) \right)^{\top}$ is the transpose of the vector $\left(1,\,-\Phi'(\bar q) \right)\in \R^2$ and "$\cdot$" the matrix-vector product.\\
Notice that $\Phi''(\bar q) < 0$ implies $C\equiv C_{\bar q, \bar m}>0$. 

\begin{prop} \label{PPP} Under the assumptions of Theorem \ref{teo3}, the point process
\beq 
\Xi_N \defeq \sum\limits_{\alpha=1}^{2^N} \delta_{ \Hm_N(\alpha) - N\left[ \Phi(\bar q) - \EE_{\nubb}( f ) \right] - \om },
\eeq
where
\beq \label{omega}
\om \defeq -\frac{1}{\bar m}\log \sqrt{ 2\pi N |\Sigma| C };
\eeq
converges weakly to a Poisson point process with intensity measure $e^{-\bar mz}dz$.
\end{prop}

Theorem \ref{teo3} follows from

\begin{itemize}
\item[i)] Proposition \ref{PPP} together with 
\item[ii)] the exponential transform 
\[
\Hm_N(\alpha) - N\left[ \Phi(\bar q) - \EE_{\nubb}( f ) \right] - \om \mapsto  \exp\left( \Hm_N(\alpha) - N\left[ \Phi(\bar q) - \EE_{\nubb}( f ) \right] - \om \right)\,,
\]
which maps the Poisson point process with intensity measure $e^{-\bar mz}dz$ to a Poisson point process on the positive line with intensity measure $t^{-\bar m - 1} dt$;
\item[iii)] the fact that infinite volume limit and the Gibbs-normalization commute.
\end{itemize}
Items {\it ii-iii)} are fairly standard in the literature: their proof is omitted, but we refer the reader to, say, \cite{BK1} for details.  \\

\noindent For the sake of simplicity, we shall prove Proposition \ref{PPP} assuming that Theorem \ref{BRR_19} (see Appendix \ref{local_CLT} and \cite[Theorem $19.5$]{BRR}) holds for the normalized vectors \eqref{vectors}; this is equivalent to the assumption that at least one measure among ${\nubb}\circ(g\circ\pi_1-\bar q)^{-1}$ and ${\nubb}\circ(f-\EE_{\nubb}(f))^{-1}$ has a density w.r.t. the Lebesgue measure on $\R$.

\begin{lem} \label{change_measure}
In low temperature, for any  $(\mu\smot\gamma)^{{\small \otimes} N}-$measurable function $F:S^{2N} \to \R$
\beq
2^N\E\left( F\left(X_{1,1},Y_{1,1}, \dots, X_{1,N},Y_{1,N}\right) \right)=\EE\left( e^{-\bar m\sqrt{N}(\Phi'(\bar q)X_N+Y_N)} \, F( \boldsymbol Z ) \right).
\eeq
\end{lem}

\begin{proof}
By definition, for any ${\bf y} = (y_1, \cdots, y_N) \in S^{2N}, y_i\in S^2$, it holds: 
\beq \bea \label{radon_G_N}
\frac{d(\mu\smot\gamma)^{{\small \otimes} N}}{d{\nubb}^{{\small \otimes} N}}({\bf y} ) & = \prod\limits_{i=1}^N \frac{d(\mu\smot\gamma)}{d{\nubb}}(y_i) = \exp\left(N\log Z^{\bar q, \bar m} -\bar m\sum\limits_{i=1}^N \Phi'(\bar q)\, g^2\circ\pi_1(y_i) + f(y_i) \right)\,, 
\eea \eeq
hence
\beq \label{after_change_F} 
2^N\E\left( F\left(X_{1,1},Y_{1,1}, \dots, X_{1,N},Y_{1,N}\right) \right)= e^{N\left(\log Z^{\bar q, \bar m}+\log2\right)} \int e^{-\bar m\sum\limits_{i=1}^N \Phi'(\bar q) g^2\circ\pi_1(y_i) + f(y_i)} F({\bf y}) d{\nubb}^{{\small \otimes} N}({\bf y}) . 
\eeq 
By the entropy condition \eqref{conditions},
\beq \label{post_radon_G_N}
\log Z^{\bar q, \bar m}+\log2 = \bar m \left[ \Phi'(\bar q) \bar q + \EE_{\nubb}(f) \right].
\eeq
Plugging this in \eqref{after_change_F}, and remembering the definition of the vectors \eqref{vectors}, the Lemma follows straightforwardly.
\end{proof}

\vskip1cm

\begin{proof}[Proof of Proposition \ref{PPP}]

We will show that for any compact $K\subset \R$ 

\beq \label{half_kallenberg}
\lim\limits_{N\to\infty} \E\left(  \Xi_N(K) \right) = \int_K e^{-\bar mz} dz\,.
\eeq
Due to the complete independence over the $\alpha$-s, this suffices to prove the Lemma by Kallenberg's theorem \cite[Theorem 4.15]{kallenberg}. To this aim, recall that, by definition,

\beq
\Hm_N(1) =  N\Phi\left(\frac{1}{N}\sum\limits_{i=1}^N g^2(X_{1,i}) \right) + \sum\limits_{i=1}^N f(X_{1,i},Y_{1,i}),
\eeq

so that

\beq \label{expt_first}
\bea
\E\left[ \Xi_N(K) \right] & =\E\left[ \sum\limits_{\alpha=1}^{2^N} \delta_{\Hm_N(\alpha)-N(\Phi(\bar q)+\EE_{\nubb}(f))-\om}(K)\right]\\
& =2^N \PP\left[\Hm_N(1)-N(\Phi(\bar q)+\EE_{\nubb}(f))-\om\in K\right]\\
& =2^N \E\left[\boldsymbol1_{\Hm_N(1)-N(\Phi(\bar q)+\EE_{\nubb}(f))-\om\in K}\right].
\eea
\eeq

By Lemma \ref{change_measure},

\beq \label{app5}
\E\left[ \Xi_N(K) \right] =  \EE\left[e^{-\bar m    \sqrt{N}(\Phi'(\bar q)X_N+Y_N)     }\boldsymbol1_{N\left[ \Phi\left( \frac{X_N}{\sqrt{N}}+\bar q\right) - \Phi(\bar q) \right] + \sqrt{N}Y_N-\om\in K}\right].
\eeq

As $\Phi$ is twice differentiable, for any fixed $N\in \N$ we can consider a map

\beq
x\in I_N \defeq \left[ - \sqrt{N},  \,\,  \sqrt{N} \right] \, \to \, \xi_N(x) \in 
I \defeq [\bar q -1, \, \bar q +1],
\eeq

that to any $x\in I_N$ associates a point $\xi_N(x)\in I$ such that

\beq
\bea
\Phi\left(\frac{x}{\sqrt{N}} +\bar q\right)-\Phi(\bar q)-\Phi'(\bar q)\frac{x}{\sqrt{N}} & = \frac{1}{2}\Phi''(\xi_N(x))\frac{x^2}{N} \\
& = R_N(x)\frac{x^2}{N},
\eea
\eeq

where to lighten notation we defined

\beq \label{def_RN}
R_N(x)\defeq\frac{1}{2}\Phi''(\xi_N(x)).
\eeq 

As $\bar q, g^2(x) \in [0,1] \,\, \forall x\in S$, it holds

\beq \label{g_bound}
X_N \in I_N \,\, \text{for any realization of} \,\, X_N;
\eeq

specifically we can write

\beq \label{shift_and_taylor}
N\left[ \Phi\left( \frac{X_N}{\sqrt{N}}+\bar q\right) - \Phi(\bar q) \right]  \equiv \sqrt{N} \Phi'(\bar q)X_N +R_N(X_N)X_N^2.
\eeq

\vskip0.5cm

Notice also (recalling that, by assumption,  $\Phi''(a)<0 \,\, \forall a \in \R$ ):

\beq \label{consequences_concavity1}
R_N(x)< 0 \quad \forall x\in I_N,
\eeq

\beq  \label{consequences_concavity3}
\lim_{N\to\infty} R_N(x) = \frac{1}{2}\Phi''(\bar q)<0 \quad \text{uniformly for} \,\, x\in o(\sqrt{N}),
\eeq

\beq  \label{consequences_concavity2}
\text{if} \qquad R \defeq \frac{1}{2}\, \inf\limits_{I} |\Phi''| > 0 \qquad \text{then} \qquad |R_N(x)| \geq R \quad \forall x\in I_N.
\eeq

\vskip0.5cm

Going back to \eqref{app5} we rewrite it as

\beq \label{num_points_int}
\E\left[ \Xi_N(K) \right] = \int_{I_N \times \R} e^{-\bar m \sqrt{N} (\Phi'(\bar q)x+y) } \Psi_N(x,y) dQ_N(x,y)
\eeq

where $Q_N \in \mathcal{M}_1^+(\R^2)$ is the distribution of $(X_N, Y_N)$ and 

\beq \label{indicator}
\Psi_N(x,y) \defeq 
\begin{cases}
1 \quad \text{if} \quad \sqrt{N}(\Phi'(\bar q)x+y)+R_N(x)x^2-\om \in K, \\
0 \quad \text{otherwise}.
\end{cases}
\eeq

It is easily seen that

\beq \label{resto}
\frac{1}{\sqrt{N}} \int_{I_N\times \R} e^{-\bar m\sqrt{N}(\Phi'(\bar q)x+y)} \Psi_N(x,y)  \,\, dxdy 
\eeq

vanishes as $N\to\infty$. Indeed, integrating by substitution according to the change of variables

\beq \label{change_variables}
(\tilde x , \,\, \tilde y) = \left(x, \,\, \sqrt{N}(\Phi'(\bar q)x+y)+R_N(x) x^2- \om \right)
\eeq

the function $\Psi_N(x,y)$ becomes $\boldsymbol 1_K(\tilde y)$ so that through \eqref{consequences_concavity1},  \eqref{consequences_concavity3} we easily obtain

\beq
\eqref{resto} \leq \frac{e^{-\bar m\om}}{N} \int_K e^{-\bar my'} dy' \int_\R e^{-\bar mRx'^2} dx' = O\left( \frac{1}{\sqrt{N}} \right) \qquad (N \uparrow \infty),
\eeq
where we also used the fact that $z\to e^{-\bar m z}$ is non-increasing and that the definition of $\om$ \eqref{omega} implies $e^{-\bar m\om}/N = O\left( N^{-1/2} \right)$. But this, through \eqref{num_points_int} and Theorem \ref{BRR_19},  yields
\beq \label{gauss_ppp}
\bea
\lim\limits_{N\to\infty} \E\left(\Xi_N(K)\right) =\lim\limits_{N\to\infty} \int_{I_N\times \R} e^{-\bar m\sqrt{N}(\Phi'(\bar q)x+y)} \Psi(x,y) \, \varphi_{\Sigma}(x,y) dxdy
\eea
\eeq
where $\varphi_{\Sigma}$ is the Gaussian bivariate density with mean $(0,0)$ and covariance matrix $\Sigma$.

We nof focus on the right hand side of \eqref{gauss_ppp} and write it as
\beq
\int_{I_N\times \R} e^{-\bar m\sqrt{N}(\Phi'(\bar q)x+y)} \Psi_N(x,y) \, \varphi_{\Sigma}(x,y) dxdy = \mathcal{J}^1_N + \mathcal{J}^2_N,
\eeq
where
\beq
\mathcal{J}^1_N\defeq \int \int_{-\log N}^{\log N}  e^{-\bar m\sqrt{N}(\Phi'(\bar q)x+y)} \Psi_N(x,y) \, \varphi_{\Sigma}(x,y) \, dx \, dy,
\eeq
and
\beq \label{gaussian_resto}
\mathcal{J}^2_N\defeq \int \int_{\left[-\log N,\, \log N \right]^{\mathsf{c}} \, \cap \, I_N} e^{-\bar m\sqrt{N}(\Phi'(\bar q)x+y)} \Psi_N(x,y) \, \varphi_{\Sigma}(x,y) \, dx \, dy.
\eeq
We begin by showing that 
\beq \label{premier}
\lim_{N\to\infty} \mathcal{J}^2_N= 0\,.
\eeq
Indeed,  by $K$-compactness  it holds that $K \subset [t,+\infty)$ for some $t\in \R$, so that for every $(x,y): \,\Psi_N(x,y)\neq 0$ the exponential term in \eqref{gaussian_resto} is bounded above by $e^{-\bar m \left(- R_N(x)x^2+\om+t\right)}$.  Specifically, again by virtue of \eqref{consequences_concavity1}, \eqref{consequences_concavity2}, we obtain
\beq 
\bea
\mathcal{J}^2_N & \leq e^{-\bar m(\om+t)} \int_{\left[-\log N,\, \log N \right]^{\mathsf{c}}} e^{\bar mR_N(x)x^2}  \varphi^1_\Sigma(x)dx \\
& \leq e^{-\bar m\left( R \log^2N + \om + t \right)} \in o(1) \quad \text{as} \quad N\to\infty,
\eea
\eeq
where $x\to\varphi^1_\Sigma(x)$ is the Gaussian univariate density of the first marginal of a bivariate Gaussian with density $(x,y)\to\varphi_\Sigma(x,y)$. This settles the claim \eqref{premier}.\\

As for $\mathcal{J}^1_N$, we have

\beq \label{before_change}
\mathcal{J}^1_N = \frac{1}{2\pi\sqrt{|\Sigma|}} \int_{-\log N}^{\log N}  \int e^{-m\sqrt{N}\left(\Phi'(\bar q)x+y\right)-\frac{1}{2}\vect\cdot\Sigma^{-1}\cdot\vect^\top} \Psi_N(\vect) \,\,d y \, dx,
\eeq
with $\vect^\top$ denoting transpose.  Note that with $(\tilde x, \, \tilde y)$ the variables as in \eqref{change_variables},  then
$$
\bea
\vect & \equiv \left( \tilde x, \, \frac{\tilde y - R_N(\tilde x) \tilde x^2 + \om}{\sqrt{N}} - \Phi'(\bar q)\tilde x \right)\\
& = \tilde x \cdot \left(1, \, -\Phi'(\bar q) \right) +  \frac{\tilde y - R_N(\tilde x) \tilde x^2 + \om}{\sqrt{N}} \cdot (0,1) 
\eea
$$
and that for any scalars $a,b\in\R$ and vectors $\vect_1, \vect_2 \in \R^2$ it trivially holds
\beq \label{distrib}
\bea
-\frac{1}{2}\left( a \vect_1+ b \vect_2 \right) \cdot \Sigma^{-1} \cdot \left( a \vect_1^\top+ b \vect_2^\top \right) = &-\frac{1}{2}a^2\vect_1 \cdot \Sigma^{-1} \cdot \vect_1^\top -ba\vect_1 \cdot \Sigma^{-1} \cdot \vect_2^\top + \\
& -\frac{1}{2}b^2\vect_2 \cdot \Sigma^{-1} \cdot \vect_2^\top.
\eea
\eeq
Specifically, as
\beq
\tilde y \in K, \,\, |\tilde x|\leq \log N \quad \Rightarrow \quad \frac{ \tilde y - R_N(\tilde x) \tilde x^2 +\om}{\sqrt{N}} \in O\left(\frac{\log^2 N}{\sqrt{N}}\right) \quad \left( N\to\infty \right)
\eeq
uniformly, and in light of \eqref{distrib}, we get that the quadratic exponent of the Gaussian density in the new variables $(\tilde x, \,\, \tilde y )$ equals
\beq
-\frac{1}{2}\tilde x^2\left(1, \, -\Phi'(\bar q) \right) \cdot \Sigma^{-1} \cdot \left(1, \, -\Phi'(\bar q) \right)^\top + O\left(\frac{\log^3 N}{\sqrt{N}}\right) \quad \left( N\to\infty \right).
\eeq

Therefore

\beq
\bea
\mathcal{J}^1_N 
= \frac{e^{-\bar m\om}}{2\pi\sqrt{|\Sigma|N}} \int_K e^{-\bar m \tilde y} \int_{-\log N}^{ \log N} e^{ - \tilde x^2\left\{ -\bar mR_N(\tilde x) + \frac{1}{2}\left(1,-\Phi'(\bar q)\right)\cdot\Sigma^{-1}\cdot\left(1,-\Phi'(\bar q)\right)^\top \right\} + o(1)}  \,\,d \tilde  x \, d \tilde y.
\eea
\eeq

As \eqref{consequences_concavity2} holds, we have

\beq 
\lim\limits_{N\to \infty} \sup_{x\in\R:\,\, |x|\leq \log N} R_N(x) =\frac{1}{2}\Phi''(\bar q),
\eeq

and by definition

\beq
\frac{e^{-\bar m\om}}{ 2\pi\sqrt{|\Sigma|N}}= \sqrt{ \frac{C}{2\pi} }.
\eeq

All in all, we get 

\beq
\lim\limits_{N\to\infty} \mathcal{J}^1_N  = \sqrt{ \frac{C}{2\pi} } \int_K e^{-\bar m\tilde y} \int e^{ - \frac{C}{2}\tilde x^2 }  \,\,d \tilde x \, d \tilde y = \int_K e^{-\bar m \tilde y} d \tilde y
\eeq

which ends the proof.

\end{proof}

\subsection{The limiting law of the overlap: proof of Proposition \ref{overlap_law}.} In this subsection we prove Proposition \ref{overlap_law}, i.e. we show that if the system is in low temperature and $\bar \nub$ is the solution of the Boltzmann-Gibbs principle \eqref{bgvp_rem_tap}, then the limits \eqref{overlap_limit_eq}, \eqref{overlap_limit_neq} hold.\\
First, we prove \eqref{overlap_limit_eq}. We start with a technical Lemma, which is a direct consequence of Lemma \ref{ub_lemma} and Theorem \ref{teo1}.\\
\begin{lem} \label{rep_av_F_as}
For every $r > 0$ and $O_N=O_N(\alpha,\alpha')$ s.t. for some $M>0$: $\left| O_N(\alpha, \alpha') \right| \leq M$ for every $\alpha, \alpha' \in \{1, \cdots, 2^N\}, \, N\in \N$; it holds
$$\lim_{N\to\infty} \left\langle {\bf 1}_{B^c_r}[\Lb] O_N(\alpha, \alpha') \right\rangle_N = 0\quad \PP-\text{a.s.}$$
where $B^c_{r} = \{\rhob \in \pst:\quad d(\rhob,\bar \nub)\geq r\}=\pst \setminus B_{\bar{\boldsymbol{\nu}},r}$.
\end{lem}
\begin{proof}
Fix $r>0$, then
$$\bea 
\left\langle {\bf 1}_{B^c_r}[\Lb] \left| O_N(\alpha, \alpha') \right| \right\rangle_N  & \leq M \left\langle {\bf 1}_{B^c_r}[\Lb] \right\rangle_N \\
& = M \sum\limits_{\alpha: \Lb\not\in B_r}\Gm \sum\limits_{\alpha'\leq 2^N} \Gmp \\
& = M \sum\limits_{\alpha: \Lb\not\in B_r}\Gm.
\eea$$
Therefore, showing that $\PP-$a.s.
\beq \label{claim4}
\lim\limits_{N\to\infty}\sum\limits_{\alpha: \Lb\not\in B_r}\Gm=\lim\limits_{N\to\infty}\frac{\sum\limits_{\alpha: \Lb\not\in B_r}e^{N\Phib[L_{N,\alpha}]}}{Z_N}=0,
\eeq
would prove the Lemma.
We therefore claim \eqref{claim4}.\\
Lemma \ref{ub_lemma} and Theorem \ref{teo1} imply that $\PP-$a.s. it holds
$$F_r \defeq \limsup\limits_{N\to \infty} \frac{1}{N}\log \sum_{\alpha: \Lb \not\in B_r} e^{N\Phib[L_{N,\alpha}]} < \lim\limits_{N\to \infty} \frac{1}{N}\log Z_N = F.$$
Specifically, for $\eta\defeq F-F_r > 0$, it holds
$$Z_N > e^{N\left(F-\frac{\eta}{4}\right)}, \quad \sum\limits_{\alpha: \Lb\not\in B_r}e^{N\Phib[L_{N,\alpha}]} < e^{N\left(F_r+\frac{\eta}{4}\right)},$$
$\PP-$a.s and provided that $N$ is large.
As the latter implies
$$\sum\limits_{\alpha: \Lb\not\in B_r}\Gm=\frac{\sum\limits_{\alpha: \Lb\not\in B_r}e^{N\Phib[L_{N,\alpha}]}}{Z_N}\leq e^{-\frac{\eta}{2}N},$$
the claim \eqref{claim4} is settled and the Lemma follows.
\end{proof}
\vskip0.4cm
\noindent Notice that the non-negative operator $T:\pst \to \R$ defined as 
$$T[\nub]\defeq \left( \int g^2(x)\nu_1(dx)-\bar q \right)^2$$ 
is continous, bounded and such that (recall $\bar q = \int g^2(x) \bar{\nu}_1$)
$$T\left[ \Lb \right]=\left(q_N(\alpha, \alpha)-\bar q\right)^2, \quad T\left[ \nubb \right]=0.$$
Specifically, given an arbitrary $\delta >0$ we can find $r>0$ such that if $\Lb\in B_r \equiv \oball$, then $T[\Lb] < \delta$. This implies
\beq \label{pr3.6_1}
\bea
\repav{{\bf 1}_{B_\epsilon}[\Lb]\left( q_N(\alpha,\alpha')-\bar q\right)^2\delta_{\alpha=\alpha'}}
&=\repav{{\bf 1}_{B_\epsilon}[\Lb]T\left[ \Lb \right]\delta_{\alpha=\alpha'}}\\
&=\E \sum\limits_{\alpha: \Lb \in B_\epsilon } T[\Lb] \Gm^2 \leq \delta
\eea \eeq
where in the last line we used that in low temperature Theorem \ref{teo3} guarantees that the process $\{ \Gm \}_{\alpha\leq 2^N}$ converges to a Poisson-Dirichlet point process with parameter $\bar m \in (0,1)$, and this implies $\E \sum_{\alpha \leq 2^N} \Gm^2 \sim 1-\bar m \in (0,1)$. Specifically, applying Lemma \ref{rep_av_F_as} we get 
$$\lim\limits_{N\to\infty} \repav{\left( q_N(\alpha,\alpha')-\bar q\right)^2\delta_{\alpha=\alpha'}} = \repav{{\bf 1}_{B_\epsilon}[\Lb]\left( q_N(\alpha,\alpha')-\bar q\right)^2\delta_{\alpha=\alpha'}} \leq \delta,$$
which, being $\delta$ arbitrarily small, proves the first claim \eqref{overlap_limit_eq} of Proposition \ref{overlap_law}.\\

In order to conclude the proof of the Proposition, we only need to show \eqref{overlap_limit_neq}.\\

To this aim, define 
$$\Tlb\defeq\frac{1}{N-2}\sum\limits_{i=3}^N\delta_{\left(X_{\alpha,i},Y_{\alpha,i}\right)},$$
and let $\tlbu, \tlbd$ be its marginals. Then, similarly to the proof of \ref{PPP}, for any fixed $N\in \N$ consider a map $\alpha \to \, \zeta_N^\alpha$ that to any configuration $\alpha$ associates a point $\zeta_N^\alpha\in [0,1]$ (specifically in the interval between $\EE_{L^1_{N,\alpha}}(g^2)$ and $\EE_{\tlbu} (g^2)$) such that

\beq
\bea
\Phi\left(\EE_{L^1_{N,\alpha}} (g^2) \right)=\Phi\left(\EE_{\tlbu} (g^2) \right)+\Phi'\left( \zeta_N^\alpha \right)\left[ \EE_{L^1_{N,\alpha}}(g^2) - \EE_{\tlbu} (g^2) \right].
\eea
\eeq
\\

By construction, one has

\beq\label{two_ham}\Hm_N(\alpha)=N\Phib\left[\Lb\right] = (N-2)\Phib\left[\Tlb\right]  + W_N{}(\alpha)+R_N(\alpha)
\eeq

where

$$\bea
W_N(\alpha)&\defeq \Phi'(\zeta^\alpha_N)\left[ g^2(X_{\alpha,1})+g^2(X_{\alpha,2})\right]+f(X_{\alpha,1},Y_{\alpha,1})+f(X_{\alpha,2},Y_{\alpha,2}) \, ;\\
R_N(\alpha)&\defeq2\Phi\left( \EE_{L^1_{N,\alpha}} (g^2)\right) -2\Phi'\left( \zeta_N^\alpha\right)\EE_{L^1_{N,\alpha}} (g^2) \, .
\eea$$
This implies that for every map $(\alpha, \alpha')\to O_N(\alpha, \alpha')$
\beq\label{confronto_av}
\bea
\E \langle O_N(\alpha,\alpha') \rangle_N & = \E \frac{\langle\langle O_N(\alpha,\alpha') e^{W_N(\alpha)+W_N(\alpha') + R_N(\alpha) +  R_N(\alpha')} \rangle\rangle_N }{\langle\langle e^{W_N(\alpha)+W_N(\alpha') + R_N(\alpha) +  R_N(\alpha')} \rangle\rangle_N }\\
\eea
\eeq
where we defined
$\langle\langle \cdot \rangle\rangle_N \defeq \tilde{Z}_N^{-2}\sum_{\alpha, \alpha'}\cdot \,\, \Gmt \Gmtp$
and
$$\Gmt \defeq \frac{\exp{(N-2)\Phib[\Tlb]}}{\tilde{Z}_N}, \quad \tilde{Z}_N
\defeq \sum\limits_{\alpha=1}^{2^N} e^{(N-2)\Phib[\Tlb]}.$$

As $\Phi$ is twice differentiable, for every $\delta>0$ there exists $\delta'>0$ such that, being $L^1_{N,\alpha}$ the first marginal of $\Lb$, if
\beq\label{cont_Phi}
\left|  \EE_{L^1_{N,\alpha}}(g^2) - \bar q\right| \leq \delta', \quad \left| \zeta_N^\alpha  - \bar q\right|  < \delta'
\eeq
then
$$\begin{cases}
\left| R_N(\alpha) - 2\Phi(\bar q) +2\Phi'(\bq)\bq\right| < \delta, \\
\left| W_N(\alpha)-\Phi'(\bq)\left[ g^2(X_{\alpha,1})+g^2(X_{\alpha,2})\right]-f(X_{\alpha,1},Y_{\alpha,1})-f(X_{\alpha,2},Y_{\alpha,2}) \right| < \delta.
\end{cases}$$
Moreover, by duality and continous projection, it is clear that we can choose $r$ small enough s.t.
\beq \label{need2}
\Tlb\in B_{\nubb, r} \,\, \Rightarrow \,\, \EE_{\tlbu}(g^2) \in \left[\EE_{\bar{\nu}_1}(g^2)-\delta, \, \EE_{\bar{\nu}_1}(g^2)+\delta\right]\equiv \left[\bq-\frac{\delta'}{2},\bq+\frac{\delta'}{2}\right].
\eeq
Notice also that 

\beq\label{interval_zeta}
\EE_{L^1_{N,\alpha}}(g^2) = \frac{N-2}{N} \left[ \frac{g^2(X_{\alpha,1})+g^2(X_{\alpha,2})}{N-2} + \EE_{\tlbu} (g^2)\right]
\eeq

which implies that as $N\to\infty$, $\EE_{L^1_{N,\alpha}}(g^2) =\zeta^\alpha_N = \EE_{\tlbu}(g^2) + O(N^{-1})$ uniformly.\\

It is readily checked that Lemma \ref{ub_lemma} and Theorem \ref{teo1} still hold if we substitute the original Hamiltonian $N\Phib[\Lb]$ with $(N-2)\Phib[\Tlb]$; this implies that Lemma \ref{rep_av_F_as} works also for the average $\langle\langle \cdot \rangle\rangle_N$. Specifically, for every $r>0$, the couples $\alpha, \alpha'$ contributing to the sums corresponding to the averages in the right hand side of \eqref{confronto_av} are the ones for which $\Tlb, \Tlbp \in  B_{\nubb, r}$. 

All in all, being $\delta$ arbitrary,  we immediately get
\beq \label{need}
\lim\limits_{N\to\infty} \E \langle F_{\alpha,\alpha'} \rangle_N = \lim\limits_{N\to\infty} \E \frac{\langle\langle \, F_{\alpha,\alpha'}\, V_{\alpha} V_{\alpha'} \, \rangle\rangle_N}{\langle\langle \, V_{\alpha} V_{\alpha'} \, \rangle\rangle_N} 
\eeq
where we defined 
$$V_{\alpha}\defeq \exp{\sum\limits_{k=1}^2 \Phi'(\bar q)g^2\left( X_{\alpha,k} \right) +f( X_{\alpha,k} , Y_{\alpha,k} )}.$$

For $O_N(\alpha,\alpha')\equiv\delta_{\alpha\neq\alpha'}\left( q_N(\alpha, \alpha') - \bar q_0 \right)^2$, the r.h.s.  of \eqref{need} equals the large-$N$ limit  of

\beq \bea \label{threebirds}
& \E\frac{\langle\langle \, q_N(\alpha,\alpha')^2\delta_{\alpha\neq\alpha'} \, V_\alpha V_{\alpha'} \, \rangle\rangle_N}{\langle\langle \, V_\alpha V_{\alpha'} \, \rangle\rangle_N} \, + \, 
\bar q_0^2 \E\frac{\langle\langle \, \delta_{\alpha\neq\alpha'} \, V_\alpha V_{\alpha'} \, \rangle\rangle_N}{\langle\langle \, V_\alpha V_{\alpha'} \, \rangle\rangle_N}-2\bar q_0\E\frac{\langle\langle \, q_N(\alpha,\alpha')\delta_{\alpha\neq\alpha'} \, V_\alpha V_{\alpha'} \, \rangle\rangle_N}{\langle\langle \, V_\alpha V_{\alpha'} \, \rangle\rangle_N}.
\eea\eeq

Since $g$ is bounded,  and by symmetry,  the first term in \eqref{threebirds} converges in the $N$-limit to 

\beq\label{quadratic}
\bea
&\frac{1}{N^2}\sum\limits_{i\neq j}\E\frac{\langle\langle \, g(X_{\alpha,i})g(X_{\alpha',i}) g(X_{\alpha,j}) g(X_{\alpha',j}) \delta_{\alpha\neq\alpha'} \, V_\alpha V_{\alpha'} \, \rangle\rangle_N}{\langle\langle \, V_\alpha V_{\alpha'} \, \rangle\rangle_N} + O\left(\frac{1}{N}\right)=\\
&=\frac{N(N-1)}{N^2}\E\frac{\langle\langle \, g(X_{\alpha,1})g(X_{\alpha',1}) g(X_{\alpha,2}) g(X_{\alpha',2}) \delta_{\alpha\neq\alpha'} \, V_\alpha V_{\alpha'} \, \rangle\rangle_N}{\langle\langle \, V_\alpha V_{\alpha'} \, \rangle\rangle_N} + O\left(\frac{1}{N}\right).
\eea\eeq
\\
Notice now that Theorem \ref{teo3} clearly also applies to the collection of weights $\Gmt,  \alpha = 1 \dots 2^N$, which therefore converges weakly to a Poisson-Dirichlet point process with parameter $\bar m$.\\
Some particularly useful properties of such process are given in  Theorem \ref{prop_ppp} of the Appendix. A simple domination argument shows that one can pass to the $N \to\infty$ limit replacing the $\{\Gmt\}_\alpha$ by the points of its weak limit so that\footnote{also using that the random vectors $(V_{\alpha}, \, g(X_{\alpha,1})g(X_{\alpha,2})V_{\alpha})$ are independent of the process $\{\Gmt\}_\alpha$.} we can use the formula \eqref{ppp_double} from Theorem \ref{prop_ppp} to get
\beq \bea 
& \lim\limits_{N\to\infty} \E\frac{\langle\langle \, q_N(\alpha,\alpha')^2\delta_{\alpha\neq\alpha'} \, V_\alpha V_{\alpha'} \, \rangle\rangle_N}{\langle\langle \, V_\alpha V_{\alpha'} \, \rangle\rangle_N} \\
& \qquad = \lim\limits_{N\to\infty}  \E \, \frac{\langle\langle \, g(X_{\alpha,1})g(X_{\alpha',1}) g(X_{\alpha,2}) g(X_{\alpha',2}) \delta_{\alpha\neq\alpha'} \, V_\alpha V_{\alpha'} \, \rangle\rangle_N}{\langle\langle \, V_\alpha V_{\alpha'} \, \rangle\rangle_N}\\
&\qquad = \lim\limits_{N\to\infty}  \E \, \frac{\sum\limits_{\alpha\neq \alpha'} \, g(X_{\alpha,1})g(X_{\alpha,2})V_\alpha \,\, g(X_{\alpha',1}) g(X_{\alpha',2}) V_{\alpha'} \,\, \Gmt\Gmtp}{\left[ \sum\limits_{\alpha} \, V_\alpha \Gmt \right]^2 }\\
&\qquad = \lim\limits_{N\to\infty} \bar m \left[ \frac{\E g(X_{1,1}) g(X_{1,2}) V_1 V_1^{\bar m-1}}{\E V_1^{\bar m}} \right]^2 = \bar m \left( \int g d\nubb \right)^4 = \bar m \bar q_0^2.
\eea\eeq
Similarly we get 

$$\bea 
\lim\limits_{N\to\infty} \E\frac{\langle\langle \, \delta_{\alpha\neq\alpha'} \, V_\alpha V_{\alpha'} \, \rangle\rangle_N}{\langle\langle \, V_\alpha V_{\alpha'} \, \rangle\rangle_N} & = \bar m, \quad \lim\limits_{N\to\infty} \E\frac{\langle\langle \, q_N(\alpha, \alpha') \delta_{\alpha\neq\alpha'} \, V_\alpha V_{\alpha'} \, \rangle\rangle_N}{\langle\langle \, V_\alpha V_{\alpha'} \, \rangle\rangle_N} = \bar m \bar q_0,& 
\eea$$

and \eqref{overlap_limit_neq} follows. This ends the proof of Proposition \ref{overlap_law}.

\appendix

\section{The relative entropy.}\label{prop_H}
\begin{Def}\label{first_def_H} Given a couple $\nu, \mu$ of Borel probability measures on a Polish space $S$, their relative entropy is defined as
\beq
H(\nu\mid\mu)=\begin{aligned}
\begin{cases}
\EE_\nu\left( \log \left( \frac{d\nu}{d\mu} \right) \right) \quad&\text{if}\quad \nu<<\mu,\quad \EE_\nu\left(\left| \log \left( \frac{d\nu}{d\mu} \right) \right| \right) < \infty,  \\
\infty\quad&\text{else}.
\end{cases}
\end{aligned}
\eeq
\end{Def}
It is a well-known fact that the Definition \ref{first_def_H} is equivalent to the following (see \cite{DV}, \cite{SYL} for details)
\beq\label{secnd_def_H} H(\nu\mid\mu)= \sup_{u \, \in \,  C_b(S)} \left\{ \int u d\nu - \log \int e^u d\mu \right\}
\eeq
and it is easily seen that $\nu \to H(\nu\mid\mu)$ from $\ps$ to $\R$ is non-negative, convex and lower semicontinous. Specifically, the sublevel $K$ defined in \eqref{Kdef} is compact in $\ps$.

\begin{Def}\label{candidates}
For all $u\in C_b(\R)$
define the measure
$$G_u\in\ps:\quad \frac{dG_u}{d\mu}(s)=\frac{e^{u(s)}}{Z_u}.$$
\end{Def}

\begin{prop}\label{ub_H_Gu}
For all $u\in C_b(S)$, $\nu\in\ps$ it holds
$$H(\nu\mid\mu)\geq H(G_u\mid\mu)+\int u \, d(\nu-G_u).$$
\end{prop}
\begin{proof}
From the variational definition of $H$ \eqref{secnd_def_H} we have
$$H(\nu\mid\mu) \geq \int u d\nu - \log Z_u \quad \forall \nu\in\ps, \forall u\in C_b(S).$$
Moreover
$$H(G_u\mid \mu) = \int \log  \left( \frac{dG_u}{d\mu} \right) dG_u = \int u dG_u - \log Z_u,$$
and the thesis follows straightforwardly.
\end{proof}

\section{Edgeworth expansions, and Poisson-Dirichlet identities.}\label{local_CLT}

Denoting by $\vect = (x,y)\to \varphi_{\Sigma}(\vect)$ the bivariate Gaussian density with mean zero and covariance matrix $\Sigma$, the following normal approximation result holds:

\begin{teor}{$\left[ \, \text{\cite{BRR}}\text{, Theorem} \,\, 19.5 \, \right]$} \label{BRR_19}
Let $\left\{ {\bf V}_i\right\}_{i\geq 1}$ be a sequence of i.i.d. random vectors with values in $\R^2$, having mean zero, a positive-definite covariance matrix $\Sigma$ and a nonzero, absolutely continous component. Then  if $\E |\!| {\bf V}_1  |\!|^3 < \infty$, writing $Q_N$ for the distribution on of $N^{-\frac{1}{2}}(  {\bf V}_1+\cdots+ {\bf V}_N)$,  one has

\beq
\int \left( 1 + |\!| x  |\!|^3 \, \right) d\left| Q_N - \Upsilon_N \right|(x,y)=o\left(\frac{1}{\sqrt{N}}\right) \quad \left( N\to\infty\right)
\eeq

where 

\beq
\frac{d\Upsilon_N}{d\lambda_2}(\vect)=\frac{h(\vect)}{\sqrt{N}}+\varphi_\Sigma(\vect)
\eeq

for a bounded smooth function $h:\R^2\to\R$, $\lambda_2$ being the Lebesgue measure on $\R^2$.

\end{teor}

\begin{teor}{$\left[ \, \text{\cite{SGCH}}\text{, Theorem} \,\, 6.4.5\, \right]$}\label{prop_ppp}
Assume that $\{v_\alpha\}_{\alpha\leq 2^N}$ is a Poisson-Dirichlet point process with intensity measure $e^{-mt}dt$ for some $m<1$, independent of a sequence $\{(U_\alpha, V_\alpha)\}_{\alpha\leq 2^N}$ of i.i.d. vectors, copies of some $(U,V)$: $\E U^2<\infty,\E V^2< \infty,  V\geq 1$. Then we have the formulas

\beq \label{ppp_single}
\E \frac{\sum\limits_{\alpha}v_\alpha U_\alpha}{\sum\limits_{\alpha}v_\alpha V_\alpha}  = \frac{\E UV^{m-1}}{\E V^m},
\eeq

\beq \label{ppp_double}
\E \frac{\sum\limits_{\alpha\neq\alpha'}v_\alpha v_\beta U_\alpha U_\beta}{\left( \sum\limits_{\alpha}v_\alpha V_\alpha\right)^2} = m\left(\frac{\E UV^{m-1}}{\E V^m}\right)^2,
\eeq

\beq \label{ppp_squared}
\E \frac{\sum\limits_{\alpha}v_\alpha^2 U_\alpha^2}{\left( v_\alpha V_\alpha \right)^2} = (1-m) \frac{\E U^2 V^{m-2}}{\E V^m}.
\eeq
\end{teor}


\begin{thebibliography}{1} 

\bibitem{AA} Arguin, Louis-Pierre, and Michael Aizenman. {\it On the structure of quasi-stationary competing particle systems.} The Annals of Probability 37.3 (2009): 1080-1113.

\bibitem{AR} Ruzmaikina, Anastasia, and Michael Aizenman. {\it Characterization of invariant measures at the leading edge for competing particle systems.} The Annals of Probability 33.1 (2005): 82-113.

\bibitem{BRR} Bhattacharya, Rabi N., and R. Ranga Rao. {\it Normal approximation and asymptotic expansions}. Society for Industrial and Applied Mathematics, (2010).

\bibitem{A} Bolthausen, Erwin, and Alain-Sol Sznitman. {\it On Ruelle's probability cascades and an abstract cavity method.} Communications in mathematical physics 197.2 (1998): 247-276.

\bibitem{B} Bolthausen, Erwin, and Alain-Sol Sznitman. {\it Ten lectures on random media.} Vol. 32. Springer Science \& Business Media, (2002).

\bibitem{BK1} Bolthausen, Erwin, and Nicola Kistler. {\it Universal structures in some mean field spin glasses and an application.} Journal of mathematical physics 49.12 (2008): 125205.

\bibitem{BK2} Bolthausen, Erwin, and Nicola Kistler. {\it A quenched large deviation principle and a Parisi formula for a Perceptron version of the GREM.} Probability in complex physical systems. Springer, Berlin, Heidelberg, (2012). 425-442.


\bibitem{CPS} Chen, Wei-Kuo, Dmitry Panchenko, and Eliran Subag. {\it The generalized TAP free energy II.} Communications in Mathematical Physics 381.1 (2021): 257-291.

\bibitem{CLR} Crisanti, Andrea, Luca Leuzzi, and Tommaso Rizzo. {\it Complexity in mean-field spin-glass models: Ising p-spin.} Physical Review B 71.9 (2005): 094202.

\bibitem{D} Derrida, Bernard. {\it Random Energy Model: An exactly solvable model of disordered systems.}, Phys. Rev. B 24 (1981)

\bibitem{Dgrem} Derrida, Bernard.  {\it A generalization of the random energy model which includes correlations between energies.} Journal de Physique Lettres 46.9 (1985): 401-407.

\bibitem{DV}  Donsker, Monroe D., and S.R. Srinivasa Varadhan. {\it Asymptotic evaluation of certain Markov process expectations for large time, I.} Communications on Pure and Applied Mathematics 28.1 (1975): 1-47.


\bibitem{WF}  Feller, William. {\it An introduction to probability theory and its applications.} Vols. I \& II, Wiley I 968 (1971).

\bibitem{g} Guerra, Francesco.  {\textit Broken replica symmetry bounds in the mean field spin glass model.} Comm. in Math. Phys. Vol. 233 (2002).

\bibitem{GS} Gufler, Stephan,  Adrien Schertzer,  and Marius A. Schmidt. {\it On concavity of TAP free energy in the SK model}.  arXiv:2209.08985 (2022). 

\bibitem{kallenberg} Kallenberg, Olav. (2017). {\it Random measures, theory and applications}. Springer International Publishing.

\bibitem{SYL} Lee, Se Yoon. {\it Gibbs sampler and coordinate ascent variational inference: A set-theoretical review.} Communications in Statistics -Theory and Methods 51.6 (2022): 1549-1568.

\bibitem{MPV} Mezard, Marc, Giorgio Parisi, and Miguel Angel Virasoro. {\it Spin glass theory and beyond: An Introduction to the Replica Method and Its Applications}. Vol. 9. World Scientific Publishing Company, (1987).

\bibitem{plefka} Plefka, Timm. {\it Convergence condition of the TAP equation for the infinite-ranged Ising spin glass model.} Journal of Physics A: Mathematical and general 15.6 (1982): 1971.

\bibitem{PoWo} Polyanskiy, Yury, and Yihong Wu. {\it Lecture notes on information theory} (2019).

\bibitem{R} Ruelle, David. {\it A mathematical reformulation of Derrida's REM and GREM.}, Comm.Math.Phys. 108 (1987).

\bibitem{sk} Sherrington, David, and Scott Kirkpatrick. {\it Solvable model of a spin-glass.} Physical review letters 35.26 (1975): 1792.

\bibitem{Nish} Nishimori, Hidetoshi. {\it Statistical Physics of Spin Glasses and Information Processing: an Introduction.} Oxford; New York: Oxford University Press (2001).

\bibitem{SGCH} Talagrand, Michel. {\it Spin Glasses: A Challenge for Mathematicians. Cavity and Mean Field Models}, Springer Verlag (2003).

\bibitem{TAP} Thouless, David J., Philip W. Anderson, and Robert G. Palmer. {\it Solution of 'solvable model of a spin glass'.} Philosophical Magazine 35.3 (1977): 593-601.




\end{thebibliography}
\end{document}